\documentclass[11pt,a4paper]{article}
\usepackage[margin=2.5cm]{geometry}
\usepackage{amsmath,amssymb,amsthm,graphicx,booktabs,xcolor}
\usepackage[colorlinks,linkcolor=black,citecolor=black,urlcolor=black]{hyperref}
\usepackage[round]{natbib}
\newtheorem{theorem}{Theorem}[section]
\newtheorem{proposition}[theorem]{Proposition}
\newtheorem{lemma}[theorem]{Lemma}
\newtheorem{corollary}[theorem]{Corollary}
\theoremstyle{definition}
\newtheorem{definition}[theorem]{Definition}
\theoremstyle{remark}
\newtheorem{remark}[theorem]{Remark}

\newcommand{\PP}{\mathbb{P}}
\newcommand{\MM}{\mathbb{M}}
\newcommand{\QQ}{\mathbb{Q}}
\newcommand{\EE}{\mathbb{E}}
\newcommand{\LP}{\mathcal{L}_\PP}
\newcommand{\gQ}{g_\QQ}

\newcommand{\one}{\mathbf{1}}
\newcommand{\sbar}{\bar s}
\newcommand{\sebar}{\bar s_e}

\title{\bf Optimal entry and exit for variance swaps:\\
closed-form rules for the perpetual contract}
\author{J.\ Maeda\thanks{%
E-mail: \texttt{jun.maeda@warwickgrad.net}. All errors are the author's own.}}
\date{This version: \today}

\begin{document}
\maketitle

\begin{abstract}
\noindent
Variance swaps are a convenient instrument for trading vega and convexity, and a
listed contract now trades on Cboe. We ask when a trader should put such a
position on and when she should take it off, and for a perpetual, continuously
settled contract we answer both in closed form: each threshold is the unique
root of a smooth-pasting equation in confluent hypergeometric functions.

Under the pricing measure the question has no content, the mark-to-market being
a martingale. Under the physical measure with a variance risk premium it becomes
meaningful, and then reduces: the accrued variance separates exactly, the
maturity, strike and costs are absorbed into a single forcing term whose sign
fixes the geometry of the exercise region, and what is left on the perpetual is
an affine reward on a CIR process, which the optimal-stopping literature already
solves.

Entering the position and exiting it are not mirror images. An exit rule follows
from the premium and the trading spread, both observable. For the short --- the
only side worth opening under the empirical sign of the premium --- an entry
rule exists only for an interval of carrying charges, and even there triggers
only deep in the upper tail of the physical law: a trader who is out of the
market pays nothing to stay out, so an operational entry rule needs a cost of
idle capital that the exit rule does not.
\end{abstract}

\noindent\textbf{Keywords:} optimal stopping; variance swap; entry and exit
timing; variance risk premium; carrying cost; Heston model; CIR process;
excessive function

\smallskip
\noindent\textbf{JEL Classification:} C61, G13

\section{Introduction}

\subsection{Variance swaps}

The payoff of a variance swap is
\begin{equation}\label{eq:payoff}
\frac{N_v}{2\Sigma}\bigl(\sigma_R^2-\Sigma^2\bigr)=N\bigl(\sigma_R^2-\Sigma^2\bigr),
\end{equation}
where $\sigma_R$ is the realized volatility over the life of the swap, $\Sigma$
the volatility strike, $T$ the maturity, and $N_v$ the vega notional: since
$\partial/\partial\sigma_R$ of \eqref{eq:payoff} equals $N_v$ at
$\sigma_R=\Sigma$, it is the profit or loss per volatility point of realized
volatility at the money. The realized variance is computed as
\begin{equation}\label{eq:realvar}
\sigma_R^2=\frac{252}{n}\sum_{i=0}^{n-1}\left(\ln\frac{S_{i+1}}{S_i}\right)^{2},
\end{equation}
with an annualisation factor of $252$, and $N=N_v/(2\Sigma)$ is the variance
notional. When \eqref{eq:payoff} is positive the seller pays the buyer that
amount, and when it is negative the buyer pays the seller its absolute value. At
inception the value of the swap is close to zero, the strike $\Sigma$ being set
so that the swap is fair.

Trading variance swaps has several attractions. The first is directness: a
trader who believes that the volatility of an index will realize above the level
the strike implies can simply buy the swap, and need not hedge delta or any
other Greeks over the life of the trade in order to monetise that view. The
second is that the risk does not disappear when the market moves away from the
initial spot level. A delta-hedged at-the-money option loses its gamma and vega
when spot jumps away; on a variance swap the dollar gamma is constant
\citep{allen2006}, so the investor continues to earn the realized volatility
wherever spot travels.

In that light a volatility swap, whose payoff
\begin{equation}\label{eq:volswap}
N(\sigma_R-\Sigma)
\end{equation}
is linear in the difference of realized and strike volatilities, might seem the
better instrument. But volatility is not additive, so its present value cannot
be split cleanly into a realized part and a future expected part. Variance is
additive, and that additivity is what makes a variance swap position easy to
manage: a holder who wishes to close out can trade a new variance swap of the
same maturity with an appropriate variance notional, up to counterparty risk and
margin.

A further reason for the popularity of variance swaps is that they are
replicable with a strip of calls and puts. The payoff defined by
\eqref{eq:payoff} and \eqref{eq:realvar} is a \emph{discrete} variance swap; the
continuous counterpart replaces \eqref{eq:realvar} by
\begin{equation}\label{eq:contvar}
\sigma_R^2=\frac1T\int_0^T\sigma^2_{S_t}\,dt .
\end{equation}
For the continuous contract,
\begin{equation}\label{eq:replication}
\EE\left[\int_0^T\sigma^2_{S_t}dt\right]
=2\left[\int_0^F\frac{dK}{K^2}\tilde P(K)+\int_F^\infty\frac{dK}{K^2}\tilde C(K)\right],
\end{equation}
where $\tilde C(K)$ and $\tilde P(K)$ are undiscounted call and put premia and
$F$ is the forward price for delivery at $T$ \citep{gatheral2006}. Perfect
replication is of course unattainable in practice: strikes are discrete, the
listed options may lack liquidity, and for some strikes the bid/offer is very
wide. Volatility swaps admit a similar but more involved replication
\citep{carr2021}.

\begin{remark}
Equation \eqref{eq:replication} is the foundation of the VIX calculation, which
selects sufficiently liquid options in the first two monthly expiries and
evaluates a discrete version of \eqref{eq:replication}; see \citet{cboe2024b}.
\end{remark}

\subsection{The Heston model}

We work throughout in the Heston model. Since the instruments we price are
contingent claims, we state it under the pricing measure $\QQ$ and carry the
subscript from the outset; Section~\ref{sec:setup} adds the corresponding
specification under the physical measure, with the same $\gamma$ but different
mean-reversion parameters. Thus
\begin{equation}\label{eq:heston}
\begin{cases}
dS_t=rS_t\,dt+\sqrt{v_t}\,S_t\,dW^{(1)}_t,\\[2pt]
dv_t=\kappa_\QQ(\theta_\QQ-v_t)\,dt+\gamma\sqrt{v_t}\,dW^{(2)}_t,
\end{cases}
\qquad \EE\bigl[dW^{(1)}_tdW^{(2)}_t\bigr]=\varrho\,dt .
\end{equation}
For the variance process to remain strictly positive we impose the Feller
condition
\begin{equation}\label{eq:feller}
2\kappa_\QQ\theta_\QQ>\gamma^2 .
\end{equation}
The model of \citet{heston1993} is popular in academia and in practice: it
prices vanilla options in closed form, the stock follows a geometric Brownian
motion, and the variance process --- a CIR process --- exhibits mean reversion.
Only the variance equation of \eqref{eq:heston} will matter below, since the
instruments we consider are functions of realized variance alone.

\subsection{Contribution}

The existing literature on variance swaps concentrates on replication
\citep{carr2021} and on the fair strike under a given model
\citep{broadie2008,bernard2014}. The trading side is much less developed. Here
we study \emph{when} to put a position on and when to take it off, in the spirit
of the optimal-trading literature for other instruments
\citep{leung2014,cartea2015,li2016}. The question has become more than academic:
Cboe listed a redesigned S\&P~500 Variance Futures contract on 23 September 2024
\citep{cboe2024a,bartholomew2024}, so a position can now be entered and unwound
on exchange, with margin, at a spread narrower than the historical OTC market
offered.

This paper solves that problem in closed form for a perpetual, continuously
settled variance swap: both the level at which the position is opened and the
level at which it is closed are unique roots of smooth-pasting equations in the
confluent hypergeometric solutions of the CIR generator. The route there is a
reduction in three steps, and the reduction rather than the machinery is what is
new. The reduced problem, an affine reward on a CIR process entered and exited
once against transaction costs, is one \citet{dayanik2003} and \citet{leung2014}
have already solved; the work is in showing that a variance swap gives rise to a
problem of that form, and in reading what the resulting thresholds say.

\begin{enumerate}\itemsep3pt
\item \textbf{The problem must be posed under $\PP$.} Under the pricing measure
it has no content: the mark-to-market is a $\QQ$-martingale
(Lemma~\ref{lem:mart}), the discount factor applied to the stopping payoff
cancels against the one inside it (Proposition~\ref{prop:cancel}), and every
stopping time attains the supremum (Corollary~\ref{cor:degen}). Content enters
only by evaluating a $\QQ$-priced reward under the trader's measure.

\item \textbf{The reduction.} Three steps take it from there to an affine
reward. The accrued variance $A_t=\int_0^tv_s\,ds$ is a second state variable,
but it separates exactly (Proposition~\ref{prop:reduce}): the value splits as
$\varepsilon A+w(t,v)$, leaving one dimension and a running reward. This needs
the absence of discounting that Proposition~\ref{prop:cancel} supplies, so
funding enters as a running charge rather than a discount rate. With an affine
variance risk premium parametrised by $\lambda=\kappa_\QQ-\kappa_\PP$, the gain
from holding an instant longer collapses to
\[
\mathcal{A}\Phi(t,v)=\varepsilon\,\lambda\,\gQ(t)\,v+(s-c_m)
\]
(Proposition~\ref{prop:carry}), absorbing the maturity, the strike and the
costs; the sign of $\varepsilon\lambda$ then fixes whether the trader unwinds
into a spike or a collapse, giving the four regimes of
Table~\ref{tab:cases}, in two of which the decaying unwind spread is a
necessity rather than a refinement. Finally, on a perpetual, continuously
settled contract the running reward comes off and the entry strike cancels,
leaving $\tilde h=\alpha+\beta v$ with
\[
\beta=-\frac{\varepsilon\lambda}{(\delta+\kappa_\QQ)(\delta+\kappa_\PP)} ,
\]
whose sign reproduces the accruing contract's geometry exactly
(Corollaries~\ref{cor:pgeom}--\ref{cor:pnoK}).

\item \textbf{The solution.} An affine reward on CIR is the setting of
\citet{leung2014}, whose argument gives the unwind threshold as the unique root
of a smooth-pasting equation (Theorem~\ref{thm:exit}), with the slope $\beta$
surviving into it where the linear-payoff case has unity. The entry obstacle
$\rho=U-\tilde h-\sbar-\sebar$ is not affine, but the entry region is disjoint
from the exercise region (Proposition~\ref{prop:interval}) and so lies where the
exit value is already known, and the same argument delivers a second
smooth-pasting equation (Lemma~\ref{lem:rho} and Theorem~\ref{thm:entry}) with
one-sidedness a conclusion rather than an assumption. The dated contract, which
has no such closed form, is characterised by a variational inequality and solved
numerically.

\item \textbf{What the solution says.} Opening the position and closing it are
not mirror images.
This part is stated for the short: under the empirical sign of the premium $D$
is unbounded above, so $c_m^\ast=-\infty$ for a long and no carrying charge
makes one worth opening (Remark~\ref{rem:long}). For the short, the exercise
region is empty, and the position held to termination, exactly when
the carrying charge falls below $c_m^\ast=\delta(\sbar-\sup_v\varepsilon D)$
(Proposition~\ref{prop:holdforever}); raising the charge lifts the unwind
threshold $b^*$ and the entry threshold $d^*$ together (Lemma~\ref{lem:mono} and
Proposition~\ref{prop:dmono}), enlarging the set of levels on which she unwinds
and shrinking the set on which she opens, so the charges admitting both a
reachable unwind and a reachable entry form an interval
(Proposition~\ref{prop:interval}). Inside it the entry threshold is still
extreme --- across the admissible charges, and across expected contract lives
from one year to thirty-two (Table~\ref{tab:horizon}), it does not fall below
about the $95$th percentile --- because
a trader who is out of the market pays nothing to stay out. Charging idle
capital $50$ basis points a year moves it to the $58$th.

\item \textbf{How much of this needs the premium.} \citet{dewbecker2017} find
the price of variance risk concentrated at one and two months and
indistinguishable from zero beyond a quarter, which is not where a perpetual
puts its weight. Appendix~\ref{app:extras} shows the construction does not
require it: a state-dependent charge $c(v)=c_m+c_1v$ leaves every step intact
and adds one term to $\beta$, and with $\lambda=0$ the two-threshold structure
survives on the slope $c_1$ alone --- at the cost of conditioning the result on
the trader's book rather than on market prices.
\end{enumerate}

\subsection{Related literature}

Our formulation under $\PP$ with a $\QQ$-priced reward is standard in the
variance-risk-premium literature \citep{carrwu2009,bollerslev2009,egloff2010},
though that literature studies static or myopic portfolio weights rather than
timing. \citet{egloff2010} is closest in spirit: they model the term structure
of variance swap rates under a two-factor affine specification and derive optimal
variance swap \emph{investments}, but the position is rebalanced continuously
rather than stopped. The reduced problem this paper arrives at --- an affine reward on a CIR process,
entered and exited once, with transaction costs --- is one the literature has
already solved. \citet{dayanik2003} give the general one-dimensional theory, and
\citet{leung2014} apply it to precisely this pair of starting and stopping
problems on CIR, obtaining both thresholds from smooth-pasting equations in the
confluent hypergeometric solutions. Sections~\ref{sec:excessive}--\ref{sec:pair}
use their machinery, and we claim no novelty for it: our contribution is the
reduction that delivers a problem of that form from a variance swap, and what
the resulting thresholds turn out to say. Two features of the reward keep the
application from being immediate --- its slope is $\beta$ rather than unity, so
the smooth-pasting equations carry a factor absent in the linear-payoff case,
and the entry obstacle is not affine at all --- and both are dealt with where
they arise. \citet{li2016} treats VIX futures under mean reversion with regime
switching. The method of
proof --- characterising the value function as the smallest excessive majorant
of the reward and reading it off the smallest concave majorant of a transformed
obstacle --- goes back to \citet{dynkin1963} and \citet{dynkin1969}.

\begin{remark}
The phrase \emph{optimal variance stopping} appears in a separate literature
\citep{gad2020} where it denotes the maximisation of
$\mathrm{Var}(X_\tau)$ over stopping times. Despite the name, that problem is
unrelated to the present one; it is time-inconsistent, whereas ours is not.
\end{remark}

\subsection{Outline}

Section~\ref{sec:degen} shows that the unwind problem posed under $\QQ$ is
degenerate, which is what forces the physical measure on the formulation.
Section~\ref{sec:setup} sets the problem up under $\PP$ and proves that the
accrued variance separates exactly, leaving one state variable.
Section~\ref{sec:carry} derives the carry identity, the single term into which
the maturity, the strike and the costs collapse, and reads the geometry of the
exercise region off its sign. Section~\ref{sec:excessive} records the
excessive-function machinery of \citet{dayanik2003} that the reduced problem
will be handed to.

Section~\ref{sec:closed} introduces the perpetual, continuously settled
contract, shows that its reduced reward is affine, and solves the exit problem
in closed form. Section~\ref{sec:pair} does the same for the entry problem,
whose obstacle is not affine, and then asks when the resulting pair is of any
use to a trader: the answer is an interval of carrying charges, inside which the
entry threshold is still extreme unless idle capital is charged.
Section~\ref{sec:numerics} puts numbers to both, and solves the dated contract
numerically for comparison. Section~\ref{sec:conclusion} concludes, and names
the one empirical tension the construction does not resolve.
Appendix~\ref{app:proofs} collects the proofs, and Appendix~\ref{app:extras}
shows that a state-dependent carrying charge reproduces the whole structure with
no variance risk premium at all.

\section{The unwind problem is degenerate under \texorpdfstring{$\QQ$}{Q}}
\label{sec:degen}

Fix a variance swap of maturity $T$ struck at $\Sigma^2$ with variance notional
$N$, and write
\begin{equation}\label{eq:A}
A_t=\int_0^t v_s\,ds,\qquad
\gQ(t)=\frac{1-e^{-\kappa_\QQ(T-t)}}{\kappa_\QQ},
\end{equation}
so that $A_t$ is the variance accrued to date and $\gQ(t)$ is the duration of
the remaining variance exposure. Under \eqref{eq:heston},
\begin{equation}\label{eq:cond}
\EE^\QQ_t\left[\int_t^T v_s\,ds\right]=\theta_\QQ(T-t)+\gQ(t)\,(v_t-\theta_\QQ),
\end{equation}
which is the conditional form of the standard expression for the fair strike
\citep{gatheral2006,broadie2008,bernard2014}. Hence the time-$t$ value of the
swap is $\Pi^T_t=e^{-r(T-t)}\tfrac{N}{T}\bigl(A_t+\Phi_0(t,v_t)\bigr)$, where
\begin{equation}\label{eq:h}
\Phi_0(t,v)=\theta_\QQ(T-t)-T\Sigma^2+\gQ(t)\,(v-\theta_\QQ)
\end{equation}
collects the terms that do not depend on the variance accrued so far.

\begin{lemma}\label{lem:mart}
Under \eqref{eq:heston}, $\bigl(A_t+\Phi_0(t,v_t)\bigr)_{t\in[0,T]}$ is a
$\QQ$-martingale.
\end{lemma}

\begin{proof}
Since $\gQ'(t)=-e^{-\kappa_\QQ(T-t)}$ and $\kappa_\QQ \gQ(t)=1-e^{-\kappa_\QQ(T-t)}$,
It\^o's formula and $dA_t=v_t\,dt$ give
\[
\begin{aligned}
dh&=\Bigl[v_t-\theta_\QQ+\gQ'(t)(v_t-\theta_\QQ)
     +\gQ(t)\kappa_\QQ(\theta_\QQ-v_t)\Bigr]dt
    +\gQ(t)\gamma\sqrt{v_t}\,dW^{(2)}_t\\
  &=(v_t-\theta_\QQ)\bigl[1+\gQ'(t)-\kappa_\QQ \gQ(t)\bigr]dt
    +\gQ(t)\gamma\sqrt{v_t}\,dW^{(2)}_t,
\end{aligned}
\]
and $1+\gQ'(t)-\kappa_\QQ \gQ(t)=1-e^{-\kappa_\QQ(T-t)}-\bigl(1-e^{-\kappa_\QQ(T-t)}\bigr)=0$.
The stochastic integral is a true martingale because $\gQ$ is bounded on $[0,T]$
and $\EE^\QQ\int_0^Tv_s\,ds<\infty$.
\end{proof}

Lemma~\ref{lem:mart} is not an accident of the Heston specification: $h$ is the
undiscounted forward value of the swap, so it must be a $\QQ$-martingale in any
arbitrage-free model. Its consequences for the timing problem are immediate.

\begin{proposition}\label{prop:cancel}
Let $\tau\le T$ be a stopping time and let the trader unwind at $\tau$ for the
prevailing mark-to-market $\Pi^T_\tau$. Then the time-$0$ value of the strategy is
\[
\EE^\QQ\bigl[e^{-r\tau}\Pi^T_\tau\bigr]
=e^{-rT}\tfrac{N}{T}\,\EE^\QQ\bigl[A_\tau+\Phi_0(\tau,v_\tau)\bigr],
\]
so the objective carries no effective discounting.
\end{proposition}

\begin{proof}
$e^{-r\tau}\Pi^T_\tau=e^{-r\tau}e^{-r(T-\tau)}\tfrac{N}{T}\bigl(A_\tau+\Phi_0(\tau,v_\tau)\bigr)
=e^{-rT}\tfrac{N}{T}\bigl(A_\tau+\Phi_0(\tau,v_\tau)\bigr)$.
\end{proof}

\begin{corollary}\label{cor:degen}
$\sup_{\tau\le T}\EE^\QQ\bigl[e^{-r\tau}\Pi^T_\tau\bigr]=e^{-rT}\tfrac{N}{T}\Phi_0(0,v_0)=\Pi^T_0$,
and the supremum is attained by every stopping time.
\end{corollary}

\begin{proof}
Combine Proposition~\ref{prop:cancel}, Lemma~\ref{lem:mart} and optional
sampling.
\end{proof}

The moral is worth stating plainly. A discount factor placed on the stopping
payoff is exactly offset by the discounting already inside the mark-to-market,
so no genuine impatience is introduced by discounting, and once it is removed
the reward is a martingale. Any non-trivial timing problem for a variance swap
must break the martingale property, and there are only a few ways to do so: a
drift wedge between the pricing and physical measures, transaction costs, a
non-linear objective, or model uncertainty. We take the first two, which are
the ones a desk actually faces, and which have the further advantage of
preserving the tractable structure of the problem.

\section{Formulation under the physical measure}
\label{sec:setup}

\subsection{Measures and the variance risk premium}

Alongside the $\QQ$-dynamics of \eqref{eq:heston}, let the variance process
satisfy
\begin{equation}\label{eq:PQ}
dv_t=\kappa_\PP(\theta_\PP-v_t)\,dt+\gamma\sqrt{v_t}\,dW^\PP_t
\end{equation}
under the trader's measure $\PP$. The two are related by the usual affine market
price of variance risk, which preserves the CIR form and the long-run product;
it is maintained throughout the body, and what depends on it is exactly
this:
\begin{equation}\label{eq:link}
\kappa_\PP\theta_\PP=\kappa_\QQ\theta_\QQ,
\qquad
\lambda:=\kappa_\QQ-\kappa_\PP .
\end{equation}
We assume \eqref{eq:feller} and its $\PP$-analogue $2\kappa_\PP\theta_\PP>\gamma^2$.
Note that $\lambda$ is a difference of mean-reversion \emph{speeds}, not of
long-run levels; it is the link that ties the two together, since with
$\kappa_\PP\theta_\PP=\kappa_\QQ\theta_\QQ$ fixed a slower pricing reversion is
the same thing as a higher pricing level. Under the link the single parameter
$\lambda$ therefore carries the premium, and $\lambda<0$ is equivalent to
$\theta_\QQ>\theta_\PP$: the empirically usual configuration, in which variance
swap strikes sit above forecast realized variance
\citep{carrwu2009,bollerslev2009}. Without the link the two are independent, and
$\lambda=0$ no longer forces $\theta_\QQ=\theta_\PP$ --- a case
Remark~\ref{rem:nolink} takes up. We do not insist
on that sign. In applications $\PP$ may equally be read as the trader's own
forecasting model rather than the objective measure, in which case $\lambda$
measures the wedge between her view and the market's, and either sign is
sensible.

\subsection{Position, costs and objective}

Let $\varepsilon=+1$ for a long variance swap and $\varepsilon=-1$ for a short
one. Both sides are carried through Section~\ref{sec:closed}, which is why the
results below are stated in $\varepsilon$ and why Table~\ref{tab:cases} has four
regimes rather than two; Section~\ref{sec:pair} is confined to the short, for
the reason Remark~\ref{rem:long} gives. Unwinding is not free: the trader crosses a spread on the strike, applied
to the vega that remains, so the cost at time $t$ is $s\,(T-t)$ with $s\ge0$
quoted in variance units per unit of remaining maturity. This is the cost
structure of the OTC market and, at a narrower level of $s$, of the listed
future. We also allow a running carrying charge $c_m$ per unit time. A positive $c_m$
is the funding cost of posted margin or an internal capital charge, and is the
case of interest for a short position. We do not restrict its sign: a negative
$c_m$ is a convenience yield, the benefit of holding a claim that pays in
exactly the states where the rest of a book does not. Nothing below uses the
sign of $c_m$ --- it enters the reward as
a constant, which is all that Lemma~\ref{lem:mono} requires --- so the results
hold for either. That it is a \emph{constant} does matter, and we keep it so
until Appendix~\ref{app:extras}, which allows $c(v)=c_m+c_1v$ and shows that the
slope $c_1$ is a second source of optionality independent of the premium. Two costs attach to the trader rather than to the position: an
entry spread $s_e\ge0$, paid when a position is put on, and an opportunity cost
$c_0\ge0$ of idle capital, charged per unit time while she holds none. Neither
appears in the unwind problem, which begins with a position already held, and we
therefore set $c_0=0$ until Section~\ref{sec:idle}. That is a modelling choice
and not an innocuous one: Section~\ref{sec:idle} shows that the entry rule is
acutely sensitive to it while the unwind rule does not involve it at all, and
this asymmetry is the main conclusion of Section~\ref{sec:pair}. The full cost
vector is therefore $(s,s_e,c_m,c_0)$, of which the unwind problem uses $s$ and
$c_m$ alone. Both spreads here are rates, because the dated swap is closed by
trading against it and the spread is crossed on the strike over the vega that
remains. The perpetual of Section~\ref{sec:closed} is closed differently, by
tear-up, and carries flat concessions $\sbar$ and $\sebar$ in their place;
Remark~\ref{rem:spreadconv} sets the two conventions side by side.

Discarding the constant $e^{-rT}N/T$ of Proposition~\ref{prop:cancel}, the
trader's unwind problem --- she already holds the position, and $\tau$ is the
time at which she closes it --- is
\begin{equation}\label{eq:obj}
\sup_{\tau\in[0,T]}\ \EE^\PP\Bigl[\varepsilon\bigl(A_\tau+\Phi_0(\tau,v_\tau)\bigr)
-s\,(T-\tau)-c_m\tau\Bigr],
\end{equation}
with $\Phi_0$ as in \eqref{eq:h}.
Note the asymmetry, which is the whole point: the reward $\Phi_0$ is a
mark-to-market and is therefore built from $\QQ$-parameters, while the
expectation in \eqref{eq:obj} is taken under $\PP$.

\subsection{The accrued variance separates}

The state of \eqref{eq:obj} is nominally $(t,A_t,v_t)$. Because $dA_t=v_t\,dt$,
$A$ cannot be treated as a constant. It can, however, be removed exactly.

Equation \eqref{eq:obj} is a single number, the value at inception. To speak of
a value function we restart the problem at an arbitrary state: for
$(t,A,v)\in[0,T]\times\mathbb{R}_+\times(0,\infty)$, let
\begin{equation}\label{eq:W}
\mathcal{V}(t,A,v)=\sup_{\tau\in[t,T]}\EE^\PP_{t,v}
\left[\varepsilon\left(A+\int_t^\tau v_u\,du+\Phi_0(\tau,v_\tau)\right)
-s(T-\tau)-c_m(\tau-t)\right],
\end{equation}
the value to a trader who arrives at time $t$ holding the position with $A$ of
variance already accrued and current variance $v$, and who chooses when to
unwind. By construction $\mathcal{V}(0,0,v_0)$ is the value of \eqref{eq:obj}.

\begin{proposition}\label{prop:reduce}
For every $(t,A,v)$, the value function \eqref{eq:W} satisfies
$\mathcal{V}(t,A,v)=\varepsilon A+w(t,v)$, where
\begin{equation}\label{eq:w}
w(t,v)=\sup_{\tau\in[t,T]}\EE^\PP_{t,v}
\left[\int_t^\tau(\varepsilon v_u-c_m)\,du+\Phi(\tau,v_\tau)\right],
\qquad
\Phi(t,v)=\varepsilon\Phi_0(t,v)-s(T-t).
\end{equation}
\end{proposition}

\begin{proof}
By Proposition~\ref{prop:cancel} no discount factor multiplies $A$ in
\eqref{eq:W}, so $\varepsilon A$ is a constant and passes outside the supremum,
leaving precisely \eqref{eq:w}.
\end{proof}

The accrued variance therefore drops out of the \emph{decision} and not merely
out of the value: the optimal $\tau$ depends on $(t,v)$ alone, so the free
boundary is a curve in the $(t,v)$ plane rather than a surface in $(t,v,A)$. What remains,
\eqref{eq:w}, is a one-dimensional optimal stopping problem in the single state
$v$, with a \emph{running reward} $\varepsilon v-c_m$ and a time-dependent
obstacle $\Phi$.

\begin{remark}\label{rem:nodiscount}
The separation depends on the absence of a discount factor, which is not an
assumption but a consequence of Proposition~\ref{prop:cancel}. It is worth
seeing why it cannot be dispensed with. Suppose a factor $e^{-\delta(\tau-t)}$
were inserted in \eqref{eq:W} and one sought a solution of the same additive
form, $\mathcal{V}(t,A,v)=a(t)A+w(t,v)$. Because $dA_t=v_t\,dt$, in the
continuation region $\mathcal{V}$ solves
\[
\partial_t\mathcal{V}+v\,\partial_A\mathcal{V}+\LP\mathcal{V}-\delta\mathcal{V}=0
\]
(we take $c_m=0$ here, which changes nothing below), and substituting the ansatz
gives
\[
\bigl[a'(t)-\delta a(t)\bigr]A
\;+\;\bigl[\partial_tw+a(t)\,v+\LP w-\delta w\bigr]\;=\;0 .
\]
Since $A$ is a free coordinate, the two brackets must vanish separately. The
first forces $a(t)=a(0)e^{\delta t}$. On the exercise boundary, however, value
matching against the payoff $\varepsilon A+\Phi(t,v)$ forces $a(t)=\varepsilon$
at every $t$ at which the boundary is non-empty. An exponential and a constant
agree only if $\delta=0$, so the ansatz fails.

The economics is more transparent than the algebra. With no discounting, a unit
of variance already accrued is worth the same whenever the position is closed,
so it cannot influence the timing. With discounting its present value is
$e^{-\delta(\tau-t)}$ per unit and so depends on \emph{when} the trader stops: a
holder sitting on a large accrued balance has a reason to realise it sooner, the
exercise boundary acquires a genuine dependence on $A$, and the state no longer
reduces. Funding should therefore be modelled as the running charge $c_m$ rather
than as a discount rate --- which is in any case the more faithful description
of a listed variance future, where the cost of the position is the financing of
variation margin and not the impatience of its holder.
\end{remark}

\section{The carry identity}
\label{sec:carry}

Write $\LP=\tfrac12\gamma^2v\,\partial_{vv}+\kappa_\PP(\theta_\PP-v)\partial_v$
for the generator of $v$ under $\PP$. In \eqref{eq:w} the value function $w$
must everywhere dominate the payoff $\Phi$ from unwinding at once; in the
variational formulation $w$ rests on $\Phi$ as on an \emph{obstacle}, which is
the American-option intrinsic value by another name. What governs where $w$
separates from that obstacle is the drift of the total worth of a trader who has
not yet unwound. Define
\begin{equation}\label{eq:Y}
Y_t:=\underbrace{\Phi(t,v_t)}_{\text{value if closed now}}
\;+\;\underbrace{\int_0^t(\varepsilon v_u-c_m)\,du}_{\text{banked while holding}} .
\end{equation}
Because $\Phi$ is $C^{1,2}$ and the integral is of finite variation, It\^o's
formula applied to the first term and ordinary differentiation to the second
give
\begin{equation}\label{eq:dY}
dY_t=\bigl[\partial_t\Phi+\LP\Phi+(\varepsilon v_t-c_m)\bigr](t,v_t)\,dt
\;+\;\Phi_v(t,v_t)\,\gamma\sqrt{v_t}\,dW^\PP_t .
\end{equation}
The bracket is the object of interest, and we name it
\begin{equation}\label{eq:APhi}
\mathcal{A}\Phi:=\partial_t\Phi+\LP\Phi+(\varepsilon v-c_m),
\end{equation}
so that $dY_t=\mathcal{A}\Phi(t,v_t)\,dt+dM_t$ with $M$ a local martingale. In
operator terms, $\partial_t+\LP$ is the generator of the space--time process
$(t,v_t)$, and the running reward is added on top of it.

Two readings follow, and both are used below. Probabilistically, taking
expectations in \eqref{eq:dY} over $[t,t+h]$,
\[
\EE^\PP_t\bigl[Y_{t+h}\bigr]-Y_t
=\EE^\PP_t\int_t^{t+h}\mathcal{A}\Phi(u,v_u)\,du
\;\approx\;\mathcal{A}\Phi(t,v_t)\,h .
\]
Now $Y_t$ is what the trader collects by unwinding at once, while
$\EE^\PP_t[Y_{t+h}]$ is what she collects by committing to unwind at $t+h$. So
$\mathcal{A}\Phi\,h$ is the expected gain from deferring by $h$: the expected
move in the mark-to-market, plus the variance accruing to her while she waits,
less the carrying cost. The running reward is part of it, and must be --- the
accrual is what distinguishes a variance swap from a plain forward, and omitting
it would misplace the boundary.

Analytically, $\mathcal{A}\Phi$ is the \emph{only} channel through which the
contract enters the free-boundary problem. The value $w$ of \eqref{eq:w}
satisfies $\min\{-\partial_tw-\LP w-(\varepsilon v-c_m),\ w-\Phi\}=0$. Writing
$u:=w-\Phi\ge0$ for the excess value and substituting $w=u+\Phi$, the first
branch becomes
\[
-\partial_tu-\LP u
-\bigl[\partial_t\Phi+\LP\Phi+\varepsilon v-c_m\bigr]
=-\partial_tu-\LP u-\mathcal{A}\Phi ,
\]
so the problem collapses to
\begin{equation}\label{eq:viu}
\min\bigl\{-\partial_tu-\LP u-\mathcal{A}\Phi,\ u\bigr\}=0
\quad\text{on }(0,T)\times(0,\infty),
\qquad u(T,\cdot)=0,
\end{equation}
whose explicit form, once $\mathcal{A}\Phi$ is computed in
Proposition~\ref{prop:carry} below, is the equation solved numerically in
Section~\ref{sec:finite}. Everything else in \eqref{eq:viu} is generic: the
operator is the $\PP$-generator of the variance alone, the obstacle is $u\ge0$,
and the terminal condition is zero. The contract therefore reaches the
free-boundary problem only through $\mathcal{A}\Phi$, and
Proposition~\ref{prop:carry} shows that the strike does not reach it at all ---
$\Sigma$ enters $\Phi$ only as the additive constant $-\varepsilon T\Sigma^2$,
which $\partial_t$ and $\LP$ both annihilate. For $\LP$ this is the work of
Proposition~\ref{prop:cancel}: with the discounting cancelled there is no term
in $u$ itself, and every term of $\LP$ differentiates, so constants are killed.
Under genuine discounting the operator would carry $+ru$, which returns
$-r\varepsilon T\Sigma^2$ and puts the strike back into the equation. A single
solve of \eqref{eq:viu} therefore serves every strike and every level of
accrued variance, which is what Section~\ref{sec:pair} needs when it poses the
entry problem for the dated contract.

\begin{proposition}\label{prop:carry}
Under \eqref{eq:PQ}--\eqref{eq:link},
\begin{equation}\label{eq:carry}
\mathcal{A}\Phi(t,v)\;=\;\varepsilon\,\lambda\,\gQ(t)\,v\;+\;(s-c_m).
\end{equation}
\end{proposition}

\begin{proof}
$\Phi$ is affine in $v$, so $\Phi_{vv}=0$ and $\Phi_v=\varepsilon \gQ(t)$, and
$\partial_t\Phi=\varepsilon\bigl[-\theta_\QQ+\gQ'(t)(v-\theta_\QQ)\bigr]+s$.
Using $\gQ'=-e^{-\kappa_\QQ(T-t)}$ and $\kappa_\QQ \gQ=1-e^{-\kappa_\QQ(T-t)}$,
\[
\mathcal{A}\Phi
=\varepsilon\Bigl[(v-\theta_\QQ)\bigl(1+\gQ'\bigr)+\kappa_\PP \gQ(\theta_\PP-v)\Bigr]+s-c_m
=\varepsilon \gQ\Bigl[(\kappa_\QQ-\kappa_\PP)(v-\theta_\QQ)
+\kappa_\PP\theta_\PP-\kappa_\PP\theta_\QQ\Bigr]+s-c_m .
\]
By \eqref{eq:link}, $\kappa_\PP\theta_\PP-\kappa_\PP\theta_\QQ
=\kappa_\QQ\theta_\QQ-\kappa_\PP\theta_\QQ=\lambda\theta_\QQ$, so the bracket
equals $\lambda(v-\theta_\QQ)+\lambda\theta_\QQ=\lambda v$.
\end{proof}

Identity \eqref{eq:carry} is the organising result of the paper. It says that
the instantaneous reward to holding the position, net of costs, is the
risk-premium carry $\lambda \gQ(t)v$ --- proportional to the \emph{current
variance level} and to the \emph{remaining variance duration} --- signed by the
side of the trade and offset by $c_m-s$. Since $\gQ(t)>0$ on $[0,T)$, the sign
of $\mathcal{A}\Phi$ in $v$ is the sign of $\varepsilon\lambda$, and the geometry
of the exercise region follows at once. Write $b_T(t)$ for the free boundary of
\eqref{eq:viu} at time $t$, the level at which the exercise region meets the
continuation region; the subscript marks it as a curve in calendar time, as
against the constant threshold the perpetual contract will produce in
Section~\ref{sec:closed}.

\begin{table}[t]\centering
\begin{tabular}{cclll}
\toprule
$\varepsilon$ & $\mathrm{sign}(\lambda)$ & continuation region & requires & interpretation \\
\midrule
$+1$ & $-$ & $\{v<b_T(t)\}$ & $s>c_m$ & long against the premium: exit into a spike\\
$-1$ & $+$ & $\{v<b_T(t)\}$ & $s>c_m$ & short against a view: exit into a spike\\
$+1$ & $+$ & $\{v>b_T(t)\}$ & $c_m>s$ & long with a view: exit into a collapse\\
$-1$ & $-$ & $\{v>b_T(t)\}$ & $c_m>s$ & short the premium: exit into a collapse\\
\bottomrule
\end{tabular}
\caption{Geometry of the exercise region implied by \eqref{eq:carry}. When
$\varepsilon\lambda<0$ the carry is negative and the only motive for delay is
the decay of the unwind spread $s(T-t)$; if in addition $s\le c_m$ then
$\mathcal{A}\Phi<0$ everywhere and $\tau^*=0$. When $\varepsilon\lambda>0$ the
carry is positive and increasing in $v$, so the position is held while variance
is elevated and closed when it falls.}
\label{tab:cases}
\end{table}

\begin{definition}\label{def:myopic}
The \emph{myopic}, or zero-carry, level is the solution
$v_c(t)$ of $\mathcal{A}\Phi(t,v)=0$,
\begin{equation}\label{eq:vc}
v_c(t)=\frac{c_m-s}{\varepsilon\lambda \gQ(t)} .
\end{equation}
\end{definition}

\begin{proposition}\label{prop:bound}
The exercise region is contained in
$\{v\le v_c(t)\}$ when $\varepsilon\lambda>0$, and in $\{v\ge v_c(t)\}$ when
$\varepsilon\lambda<0$.
\end{proposition}

\begin{proof}
Where $\mathcal{A}\Phi>0$ the obstacle is a strict subsolution of the
variational inequality, so it is not optimal to stop there; by
\eqref{eq:carry} that region is $\{v>v_c(t)\}$ when $\varepsilon\lambda>0$ and
$\{v<v_c(t)\}$ when $\varepsilon\lambda<0$.
\end{proof}

The myopic level is the natural benchmark against which to read the free
boundary $b_T(t)$: it is the level at which a trader who ignores the option to
wait would act, and the gap $|b_T(t)-v_c(t)|$ is the value of that option,
expressed in the units of the state.

\section{Optimal stopping via excessive functions}
\label{sec:excessive}

We recall the tools used for the time-homogeneous problems of
Sections~\ref{sec:closed} and \ref{sec:pair}. They are not new: the general
one-dimensional theory is \citet{dayanik2003}, and the form used here, for a CIR
process entered and exited once against transaction costs, is the one
\citet{leung2014} develop. What the previous two sections have established is
that a variance swap gives rise to a problem this apparatus fits. Consider
\begin{equation}\label{eq:general}
V(x)=\sup_\tau\EE_x\bigl[e^{-\delta\tau}h(X_\tau)\bigr]
\end{equation}
for a one-dimensional regular diffusion $X$ on an interval with endpoints
$\underline{x}<\overline{x}$ (in our application $X=v$ and the endpoints are $0$
and $\infty$), and a constant $\delta>0$ whose meaning in our application is
fixed in Section~\ref{sec:closed}: it is not a discount rate but the rate at
which the contract terminates.

\begin{theorem}[\citealp{dynkin1963}]\label{thm:dynkin}
If $h$ is lower semi-continuous, the value function $V$ of \eqref{eq:general} is
the smallest $\delta$-excessive majorant of $h$, where $f\ge0$ is
$\delta$-excessive for $X$ when $f(x)\ge\EE_x\bigl[e^{-\delta\tau}f(X_\tau)\bigr]$
for all $x$ and all stopping times $\tau$.
\end{theorem}

For Brownian motion and $\delta=0$ the $0$-excessive functions are exactly the
non-negative concave ones \citep{dynkin1969}, so $V$ is the smallest non-negative
concave majorant of $h$. \citet{dayanik2003} extend this to a general
one-dimensional diffusion and $\delta>0$ by a change of scale, and this is the
form we use. Let $F$ and $G$ denote the increasing and decreasing positive
solutions of $\mathcal{L}u=\delta u$ and put
\begin{equation}\label{eq:psiphi}
\psi=\frac{F}{G},\qquad \varphi=-\frac1\psi=-\frac{G}{F},
\end{equation}
both strictly increasing.

\begin{proposition}[\citealp{dayanik2003}]\label{prop:dk}
Let $\widehat H$ be the smallest non-negative concave majorant of
$H:=(h/G)\circ\psi^{-1}$ on $[\psi(\underline{x}),\psi(\overline{x})]$. Then
$V(x)=G(x)\widehat H(\psi(x))$.
Setting $\mathcal{E}=\{x:V(x)=h(x)\}$, the time $\tau^*=\inf\{t\ge0:X_t\in\mathcal{E}\}$ is
optimal whenever $h$ is continuous.
\end{proposition}

\begin{proposition}\label{prop:dk2}
Equivalently, with $H:=(h/F)\circ\varphi^{-1}$ and $\widehat H$ the smallest
\emph{decreasing} concave majorant of $H$, $V(x)=F(x)\widehat H(\varphi(x))$.
\end{proposition}

Propositions~\ref{prop:dk} and \ref{prop:dk2} are the same statement written in
the two natural transformations; which one is convenient depends on whether the
exercise region abuts the upper or the lower endpoint, and we shall use both.

For the CIR process under $\PP$ write
\begin{equation}\label{eq:abz}
\mathsf{a}:=\frac{\delta}{\kappa_\PP},\qquad
\nu:=\frac{2\kappa_\PP\theta_\PP}{\gamma^2},\qquad
\varsigma:=\frac{2\kappa_\PP}{\gamma^2},
\end{equation}
so that $\nu$ is the shape parameter of the stationary law and \eqref{eq:feller}
reads $\nu>1$ under $\QQ$. The relevant pair is then
\begin{equation}\label{eq:FG}
F(v)=M(\mathsf{a},\nu;\varsigma v),
\qquad
G(v)=\mathcal{U}(\mathsf{a},\nu;\varsigma v),
\end{equation}
where $M$ and $\mathcal{U}$ are the confluent hypergeometric functions of the
first and second kind,
\begin{equation}\label{eq:MU}
\begin{aligned}
M(p,q;z)&=\sum_{n\ge0}\frac{p_nz^n}{q_n\,n!},
\qquad p_0=1,\quad p_n=p(p+1)\cdots(p+n-1),\\[2pt]
\mathcal{U}(p,q;z)&=\frac{\Gamma(1-q)}{\Gamma(p-q+1)}\,M(p,q;z)
+\frac{\Gamma(q-1)}{\Gamma(p)}\,z^{1-q}M(p-q+1,2-q;z).
\end{aligned}
\end{equation}
$F$ is strictly increasing and $G$ strictly decreasing \citep{going2003}, and we
shall use the derivative identities
\begin{equation}\label{eq:derivs}
\frac{d}{dz}M(p,q;z)=\frac pq M(p+1,q+1;z),
\qquad
\frac{d}{dz}\mathcal{U}(p,q;z)=-p\,\mathcal{U}(p+1,q+1;z)
\end{equation}
\citep{buchholz1969}. Under \eqref{eq:feller} the endpoint $0$ is an entrance
boundary and $+\infty$ is natural, so the state space is $(0,\infty)$ and
Proposition~\ref{prop:dk} applies with $\underline{x}=0$ and $\overline{x}=\infty$.

\section{A perpetual variance swap}
\label{sec:closed}

The excessive-function method requires a time-homogeneous, one-dimensional
diffusion, and the fixed-maturity contract of Section~\ref{sec:setup} is
neither: $\gQ(t)$ and the residual maturity in the cost both depend on $t$. The
natural time-homogeneous variance swap is a \emph{perpetual} one --- perpetual
in the sense of carrying no deterministic maturity, so that nothing about it
shortens with calendar time --- and this section solves the unwind problem for
it in closed form.

\subsection{The contract}

Let the contract pay its holder $(v_t-K)\,dt$, per unit of variance notional,
until an independent exponential time $T_\delta$ with rate $\delta>0$: a
floating-for-fixed exchange of instantaneous variance against a fixed rate,
settled continuously, and terminating at a date that is not known in advance.
It is a variance swap in the only sense that matters here: the holder accrues
realized variance while the position is open. Randomised maturities
of this kind are a standard device \citep{carr1998}; what matters here is that
$T_\delta$ is memoryless, so the contract looks the same at every date at which
it is still alive, which is what makes the problem time-homogeneous.

The parameter $\delta$ is therefore a feature of the position rather than a
description of the trader's patience: $1/\delta$ is its expected life, and the
state sensitivity computed below puts that on the same axis as the maturity of a
dated swap. What is fixed exogenously is the intensity $\delta$ and not the date
$T_\delta$, in the way that a barrier level is contractual while the time at
which it is struck is not. What ends the position admits two readings, and the
analysis uses only that $T_\delta$ is exponential and independent of $v$. Either
termination is a term of the contract, a swap running until the first arrival of
an independent Poisson process of intensity $\delta$, which is the randomisation
device of \citet{carr1998}; or the position is ended by events outside the trade
altogether --- a mandate that lapses, a desk reallocated, a counterparty
withdrawn, a roll that fails. The second needs no unusual term sheet and is the
better description of why a variance position is so often closed for reasons
having nothing to do with variance; it also supplies, in
Section~\ref{sec:pair}, the one thing the entry problem needs and a term sheet
would not give, namely an event that ends the opportunity to trade and not
merely the contract. Since the trader cannot hold the contract past $T_\delta$, and $T_\delta$ is
independent of $v$, every expectation below carries the survival probability
$\PP(T_\delta>u)=e^{-\delta u}$ as a weight. Under $\QQ$, using
$\EE\int_0^{T_\delta}\!\!X_u\,du=\int_0^\infty\! e^{-\delta u}\EE[X_u]\,du$,
\begin{equation}\label{eq:pvsvalue}
\Pi(v)=\EE^\QQ\!\left[\int_0^{T_\delta}\!(v_u-K)\,du\,\Big|\,v_0=v\right]
=\int_0^\infty\! e^{-\delta u}\,\EE^\QQ_v[v_u-K]\,du
=\frac{\theta_\QQ-K}{\delta}+\frac{v-\theta_\QQ}{\delta+\kappa_\QQ}.
\end{equation}
The first term of that integral recurs throughout, under both measures, so we
name it: for $\MM\in\{\PP,\QQ\}$,
\begin{equation}\label{eq:perp}
P_\MM(v):=\EE^\MM_v\!\int_0^\infty\! e^{-\delta u}v_u\,du
=\frac{\theta_\MM}{\delta}+\frac{v-\theta_\MM}{\delta+\kappa_\MM},
\end{equation}
the survival-weighted perpetuity of variance, affine in $v$ with slope
$1/(\delta+\kappa_\MM)$, so that $\Pi=P_\QQ-K/\delta$. The rate at which the
contract is fair is
\begin{equation}\label{eq:pvsrate}
K(v)=\theta_\QQ+\frac{\delta\,(v-\theta_\QQ)}{\delta+\kappa_\QQ}
=\delta\int_0^\infty e^{-\delta u}\,\EE^\QQ_t\bigl[v_{t+u}\bigr]\,du\Big|_{v_t=v},
\end{equation}
an exponentially weighted average of expected future variance, where a
fixed-tenor swap rate is a uniformly weighted one. The weighting is the whole of
the difference, and it is governed by the expected life: $\delta\to0$ is a
contract that never terminates, and gives $K\to\theta_\QQ$, the long-run level,
with no dependence on the current state at all, while $\delta\to\infty$ is one
that terminates at once, and gives $K\to v$. The sensitivity
$\partial K/\partial v=\delta/(\delta+\kappa_\QQ)$ equals $0.20$ at the
parameters of Section~\ref{sec:numerics}, while by \eqref{eq:cond} a dated swap
of tenor $\mathsf{T}$ has sensitivity
$(1-e^{-\kappa_\QQ\mathsf{T}})/(\kappa_\QQ\mathsf{T})$; the
two agree at $\mathsf{T}=2.5$ years. The $\delta=0.5$ contract therefore responds to
the state like a dated swap of about two and a half years, against an expected
life of two, and it is that rough agreement between two independent ways of
measuring its horizon which puts $1/\delta$ on the same axis as a
maturity.\footnote{The two need not agree exactly, since they average expected
future variance differently: the perpetual with the exponential kernel
$\delta e^{-\delta u}$, a dated swap uniformly over $[0,\tau]$. Matching the means
of the two kernels rather than their sensitivities would give $\tau=2/\delta$.
The matching tenor also depends on $\kappa_\QQ$, falling from $3.2$ to $2.2$
years as $\kappa_\QQ$ rises from $0.5$ to $5$.} That sensitivity is the same at every level, so
$K$ is unbounded in $v$: a variance shock is passed into the rate at a fixed
fraction for as long as the contract survives. This is a property of the affine
dynamics rather than of perpetual variance: under a model whose variance of
variance grows faster than affine, the same sensitivity decays with the level.

Unwinding at $\tau$ means terminating the contract with the counterparty for a
cash payment --- a tear-up, which is how an over-the-counter position without a
listed offset is in practice closed. The fair payment is the mark-to-market
\eqref{eq:pvsvalue} at the prevailing state, and by \eqref{eq:pvsrate}
\[
\varepsilon\,\Pi(v_\tau;K)
=\varepsilon\Bigl[\frac{\theta_\QQ-K}{\delta}+\frac{v_\tau-\theta_\QQ}{\delta+\kappa_\QQ}\Bigr]
=\varepsilon\,\frac{K(v_\tau)-K}{\delta},
\]
so the position closes at a number depending on the state and not on the date.
The same figure would be reached by entering an offsetting contract on the same
termination event and letting the two legs run off together, since the residual
annuity has mean length $1/\delta$ whatever $\tau$ is; but that construction
needs a counterparty willing to write a contract on this contract's own
termination clock, and the tear-up needs no such thing.

The cost of unwinding follows from the mechanism. A tear-up is a single
negotiated concession on the termination payment, so we charge a flat
$\sbar\ge0$ in variance units, and correspondingly a flat $\sebar\ge0$ on
entry in Section~\ref{sec:pair}. This is not the convention of
Section~\ref{sec:setup}, where the dated contract is closed by trading against
it and the spread is crossed on the strike, applied to the vega that remains, at
a cost $s(T-t)$ that shrinks as the swap runs off. The two are different
quantities --- $s$ a rate, $\sbar$ a total --- because they price different
transactions, and Remark~\ref{rem:spreadconv} records what is at stake in the
choice.

\begin{remark}[what $\delta$ is, and what it is not]\label{rem:whatdelta}
The factor $e^{-\delta u}$ that appears throughout this section and the next is
a survival probability, not a discount factor, and the distinction is the one
Remark~\ref{rem:nodiscount} turns on. There the obstruction to discounting was
that the accrued balance $A_\tau$ would acquire a present value depending on
\emph{when} the position is closed, so that the state would not reduce. Here
there is no accrued balance to carry: the contract settles continuously, so
variance is paid out as it is earned and nothing is held to a maturity date.
The two sections are therefore consistent, and deliberately so --- funding still
enters as the running charge $c_m$, exactly as Remark~\ref{rem:nodiscount}
requires, and $\delta$ is a property of the instrument rather than of its
holder. Writing $\EE\int_0^{T_\delta}X_u\,du=\int_0^\infty e^{-\delta u}\EE[X_u]\,du$
and $\EE[\one\{T_\delta>\tau\}Z_\tau]=\EE[e^{-\delta\tau}Z_\tau]$ turns every
expression below into one that \emph{looks} discounted; no impatience has been
introduced, and the value of a unit of variance is the same whenever it is
earned.
\end{remark}

\begin{remark}[the two spread conventions]\label{rem:spreadconv}
The dated contract of Section~\ref{sec:setup} is closed by trading against it,
so its spread is crossed on the strike and applies to the vega that remains: a
rate $s$, costing $s(T-t)$. The perpetual is closed by tear-up, a single
negotiated concession with no residual position to carry, so its cost is the
flat $\sbar$. The two are different quantities and we keep separate symbols for
them; taking $\sbar=s$, as Section~\ref{sec:numerics} does, says that closing a
perpetual costs what closing a one-year dated swap costs.

The alternative is to insist on one rate, charging the perpetual
$\sbar=s/\delta$ --- the spread run over the contract's expected remaining life.
That is the right convention if the position is unwound by holding an offsetting
contract to termination rather than torn up, and it is the more conservative
assumption, roughly doubling the round trip at $\delta=0.5$. It costs the
results a good deal and changes none of them qualitatively. The window of
Figure~\ref{fig:window} then opens at a premium of $2.15$ volatility points
rather than $1.09$ and is $35$ basis points wide at three points rather than
$110$; the entry floor of Table~\ref{tab:horizon} is unchanged at the $94.7$th
percentile for the one-year contract but leaves the support beyond four years,
so the monotone deterioration in the contract's life is sharper, not weaker; the
round trips roughly double, to $7.9$ years' wait
and $3.5$ years' hold for the short; and the state-dependent charge of
Appendix~\ref{app:extras} needs a slope of $c_1=-0.5$ rather than $-0.2$,
a swing of some $550$ basis points a year across the middle $98\%$ of the law
rather than $200$, before the two-threshold structure appears. Every conclusion of
Sections~\ref{sec:closed}--\ref{sec:pair} survives the substitution; what it
costs is the size of the premium and of the hedging demand the model needs
before a round trip exists.
\end{remark}

\begin{remark}[on the instrument]\label{rem:instrument}
No exchange lists a perpetual variance swap on equity indices, so the contract
is an idealisation. It is, however, an idealisation of the right kind. First, it
is a randomised-maturity version of the contract of Section~\ref{sec:setup}
rather than a different exposure: the holder accrues realized variance, and it is by accruing, not by
marking, that the premium of Section~\ref{sec:carry} is earned. A claim written
instead on a constant-maturity swap \emph{rate} would have no accrual at all,
and would therefore not be a variance swap in any useful sense --- selling a
one-year swap at $t=0$ and buying a one-year swap at $t=\tfrac12$ leaves a
calendar spread on $[\tfrac12,\tfrac32]$, not a flat book. Second, perpetual
contracts are by now a studied object \citep{angeris2022,he2022,ackerer2026},
and a power perpetual on $\mathrm{ETH}^2$ trades in decentralised markets whose
funding leg its issuer describes as the expected variance of the underlying
\citep{opyn2022}. One difference should be kept in view: in that literature the
funding rate is an endogenous mechanism tying the mark to an index, whereas
$\delta$ here is an exogenous intensity, fixed independently of the state and of
the trade.
\end{remark}

\subsection{Reduction to an affine reward}

The trader holds $\varepsilon$ units struck at $K$, accrues
$\varepsilon(v_t-K)-c_m$ per unit time for as long as the contract survives, and
if she unwinds before it terminates receives $\varepsilon(K(v_\tau)-K)/\delta$,
less the tear-up concession $\sbar$:
\begin{equation}\label{eq:pvsobjko}
V(v;K)=\sup_\tau\ \EE^\PP_v\!\left[\int_0^{\tau\wedge T_\delta}
\bigl(\varepsilon(v_u-K)-c_m\bigr)du
+\one\{T_\delta>\tau\}\Bigl(\varepsilon\frac{K(v_\tau)-K}{\delta}-\sbar\Bigr)\right],
\end{equation}
which, on integrating out the independent $T_\delta$, is
\begin{equation}\label{eq:pvsobj}
V(v;K)=\sup_{\tau}\ \EE^\PP_v\!\left[\int_0^\tau e^{-\delta u}
\bigl(\varepsilon(v_u-K)-c_m\bigr)du
+e^{-\delta\tau}\Bigl(\varepsilon\frac{K(v_\tau)-K}{\delta}-\sbar\Bigr)\right].
\end{equation}
We work with \eqref{eq:pvsobj} throughout, reading its exponentials as survival
probabilities in the sense of Remark~\ref{rem:whatdelta}.  The running reward is
removed in the usual way: one adds and subtracts the value of holding the
position to termination, which is a fixed function of the state carrying no
optimisation of its own and therefore does not compete with the choice of
$\tau$.

\begin{lemma}\label{lem:reduce}
Let $f$ be the running reward of \eqref{eq:pvsobj}, $h(v)=\varepsilon(K(v)-K)/\delta-\sbar$
the stopping reward, and $R(v)=\EE^\PP_v\bigl[\int_0^\infty e^{-\delta u}f(v_u)du\bigr]$
the value of holding the position until the contract terminates. Then
\[
V(v;K)=R(v;K)+U(v),\qquad
U(v):=\sup_\tau\EE^\PP_v\bigl[e^{-\delta\tau}\tilde h(v_\tau)\bigr],\qquad
\tilde h:=h-R .
\]
\end{lemma}

\begin{proof}
Fix a stopping time $\tau$ and split $\int_0^\tau=\int_0^\infty-\int_\tau^\infty$.
Substituting $u=\tau+r$ in the tail and using the strong Markov property,
\[
\EE^\PP_v\Bigl[\int_\tau^\infty e^{-\delta u}f(v_u)\,du\Bigr]
=\EE^\PP_v\Bigl[e^{-\delta\tau}\,\EE^\PP_{v_\tau}\!\int_0^\infty e^{-\delta r}f(v_r)\,dr\Bigr]
=\EE^\PP_v\bigl[e^{-\delta\tau}R(v_\tau)\bigr],
\]
all terms being finite because $f$ is affine in $v$ and $\sup_u\EE^\PP_v[v_u]<\infty$.
The value of unwinding at that $\tau$ is therefore
\[
R(v)-\EE^\PP_v\bigl[e^{-\delta\tau}R(v_\tau)\bigr]
+\EE^\PP_v\bigl[e^{-\delta\tau}h(v_\tau)\bigr]
=R(v)+\EE^\PP_v\bigl[e^{-\delta\tau}\tilde h(v_\tau)\bigr].
\]
The first term does not depend on $\tau$. It is an additive constant in the
optimisation --- exactly as $\varepsilon A$ is in Proposition~\ref{prop:reduce}
--- and passes outside the supremum, so a single supremum over $\tau$ remains and
$V=R+U$. On $\{\tau=\infty\}$ both $e^{-\delta\tau}$ and the residual integral
vanish, and the identity holds there with the convention
$e^{-\delta\tau}\tilde h(v_\tau)=0$.
\end{proof}

Both $R$ and $h$ are affine, and the entry strike cancels between them.

\begin{proposition}\label{prop:pcarry}
$\tilde h(v)=\alpha+\beta v$, where
\begin{equation}\label{eq:pbeta}
\beta=-\frac{\varepsilon\lambda}{(\delta+\kappa_\QQ)(\delta+\kappa_\PP)},
\qquad
\alpha=\varepsilon\left[\frac{\theta_\QQ-\theta_\PP}{\delta}
-\frac{\theta_\QQ}{\delta+\kappa_\QQ}+\frac{\theta_\PP}{\delta+\kappa_\PP}\right]
-\sbar+\frac{c_m}{\delta}.
\end{equation}
In particular $\tilde h$ does not depend on $K$.
\end{proposition}

\begin{proof}
Computing $R$ under $\PP$ exactly as \eqref{eq:pvsvalue} was computed under
$\QQ$,
$R(v;K)=\varepsilon\bigl[(\theta_\PP-K)/\delta+(v-\theta_\PP)/(\delta+\kappa_\PP)\bigr]-c_m/\delta$,
while $h(v)=\varepsilon\bigl[(\theta_\QQ-K)/\delta+(v-\theta_\QQ)/(\delta+\kappa_\QQ)\bigr]-\sbar$.
The terms in $K$ are identical and cancel in $h-R$; collecting the rest gives
$\alpha$, and the coefficient of $v$ is
$\varepsilon\bigl[(\delta+\kappa_\PP)^{-1}\cdot(-1)+(\delta+\kappa_\QQ)^{-1}\bigr]
=-\varepsilon\lambda/\bigl((\delta+\kappa_\QQ)(\delta+\kappa_\PP)\bigr)$.
\end{proof}

The difference of the two perpetuities in that proof is the object the rest of
the paper works with. Write
\begin{equation}\label{eq:Ddef}
D(v):=P_\QQ(v)-P_\PP(v)
     =\int_0^\infty e^{-\delta u}
      \bigl(\EE^\QQ_v[v_u]-\EE^\PP_v[v_u]\bigr)\,du
     =\frac{K(v)}{\delta}-P_\PP(v)
\end{equation}
for the premium embedded in the perpetual rate: the $\QQ$-value at which the
position can be unwound, less the $\PP$-variance it accrues if it is held. In
this notation Proposition~\ref{prop:pcarry} reads
\begin{equation}\label{eq:htD}
\tilde h=\varepsilon D-\sbar+\frac{c_m}{\delta},
\end{equation}
which is the form used from here on: the premium enters through $D$, and the
spread and the carrying charge only as constants. Under Heston $D$ is affine
with $D'=-\lambda/\bigl((\delta+\kappa_\QQ)(\delta+\kappa_\PP)\bigr)$, so
$\beta=\varepsilon D'$. The representation \eqref{eq:Ddef} uses no property of
the variance dynamics beyond integrability; only the shape of $D$ is
model-dependent.

Three consequences, and they are the reason the section is here.

\begin{corollary}\label{cor:pgeom}
$\mathrm{sign}(\beta)=\mathrm{sign}(-\varepsilon\lambda)$, so the perpetual
contract reproduces the exercise geometry of Table~\ref{tab:cases}: the position
is unwound into a volatility spike when $\varepsilon\lambda<0$ and into a
collapse when $\varepsilon\lambda>0$.
\end{corollary}

\begin{corollary}\label{cor:plam}
If $\lambda=0$ then $\beta=0$, so $\tilde h$ is constant in $v$ and the timing
problem disappears: a constant obstacle is taken at once when it is positive and
never when it is negative. Under \eqref{eq:link} that constant is
$-\sbar+c_m/\delta$.
\end{corollary}

\begin{remark}[without the affine link]\label{rem:nolink}
Corollary~\ref{cor:plam} does not need \eqref{eq:link}. Freed of the link,
$\lambda=0$ constrains only the speeds and still permits
$\theta_\QQ\neq\theta_\PP$; $D$ then collapses to the constant
$\kappa(\theta_\QQ-\theta_\PP)/\bigl(\delta(\delta+\kappa)\bigr)$ with
$\kappa=\kappa_\PP=\kappa_\QQ$, which moves $\alpha$ but not $\beta$ and so
leaves the conclusion intact.
\end{remark}

What Corollary~\ref{cor:plam} isolates is that optionality comes from the \emph{slope} of
$D$ and not from its level: a premium that does not vary with the state enters
$\tilde h$ exactly as $c_m$ does, and neither gives any reason to prefer one
moment to another. As in Section~\ref{sec:carry}, the premium is in this sense
the sole source of optionality --- a conclusion that depends on the carrying
charge being constant in $v$, and that Appendix~\ref{app:extras} withdraws when
it is not. Here it enters through the gap between the $\QQ$-rate at which the
position is unwound and the $\PP$-variance accrued while it is held, and
deleting the accrual would delete $\lambda$ from \eqref{eq:pbeta} altogether.

\begin{corollary}\label{cor:pnoK}
The exercise region does not depend on the rate $K$ at which the position was
entered, and neither, therefore, does its boundary.
\end{corollary}

\begin{proof}
By Theorem~\ref{thm:dynkin} and Proposition~\ref{prop:dk} the region is
determined by $\tilde h$ and the law of $v$ alone, and $\tilde h$ is free of $K$
by Proposition~\ref{prop:pcarry}.
\end{proof}

The entry rate therefore fixes the \emph{value} of the position, through
$R(v;K)$, and not the decision --- the same separation of value from timing that
Proposition~\ref{prop:reduce} produced for the accrued variance.

\subsection{The exit threshold}

With $\tilde h$ affine, Propositions~\ref{prop:dk} and \ref{prop:dk2} apply
directly. The following lemma records the properties of the transformed reward;
note the factor $\tilde h'=\beta$, which is absent when $h(x)=x-c$ as in
\citet{leung2014} but not here.

\begin{lemma}\label{lem:H}
Let $H=(\tilde h/F)\circ\varphi^{-1}$ with $\beta>0$. Then $H$ is continuous on
$[\varphi(0),0]$, twice differentiable on $(\varphi(0),0)$, and
\begin{equation}\label{eq:Hprime}
H'(z)=\frac{1}{\varphi'(y)}\cdot\frac{\tilde h'(y)F(y)-\tilde h(y)F'(y)}{F^2(y)},
\qquad
H''(z)=\frac{2\,(\mathcal{L}-\delta)\tilde h(y)}{\gamma^2y\,F(y)\,\varphi'(y)^2},
\qquad z=\varphi(y).
\end{equation}
Moreover $H(0)=0$; $H<0$ on $[\varphi(0),\varphi(v_0^*))$ and $H>0$ on
$(\varphi(v_0^*),0)$, where $\tilde h(v_0^*)=0$; and $H$ is convex on
$(\varphi(0),\varphi(v_s^*)]$ and concave on $[\varphi(v_s^*),0)$, where
$(\mathcal{L}-\delta)\tilde h(v_s^*)=0$.
\end{lemma}

\begin{proof}
See Appendix~\ref{app:proofs}.
\end{proof}

\begin{theorem}\label{thm:exit}
Suppose the exercise region is non-empty.
\begin{enumerate}
\item[(i)] If $\beta>0$ (equivalently $\varepsilon\lambda<0$) the continuation
region is $(0,b^*)$, $U(v)=\tilde h(b^*)F(v)/F(b^*)$ on it, and $b^*$ is the
unique root of
\begin{equation}\label{eq:spF}
\beta\,F(b)-\tilde h(b)\,F'(b)=0 .
\end{equation}
\item[(ii)] If $\beta<0$ the continuation region is $(b^*,\infty)$,
$U(v)=\tilde h(b^*)G(v)/G(b^*)$ on it, and $b^*$ is the unique root of
\begin{equation}\label{eq:spG}
\beta\,G(b)-\tilde h(b)\,G'(b)=0 .
\end{equation}
\end{enumerate}
In either case the first entry time into the exercise region is optimal.
In either case $V(v;K)=R(v;K)+U(v)$.
\end{theorem}

\begin{proof}
See Appendix~\ref{app:proofs}.
\end{proof}

Two features differ from the linear-reward case usually quoted. The
smooth-pasting conditions carry the factor $\beta=\tilde h'$; without it the
equation is not even dimensionally consistent, since $F$ and $\tilde hF'$ then
differ by a unit of variance. And which of $F$ or $G$ enters is not a modelling
choice but is fixed by $\mathrm{sign}(\varepsilon\lambda)$.

\begin{figure}[t]\centering
\includegraphics[width=\textwidth]{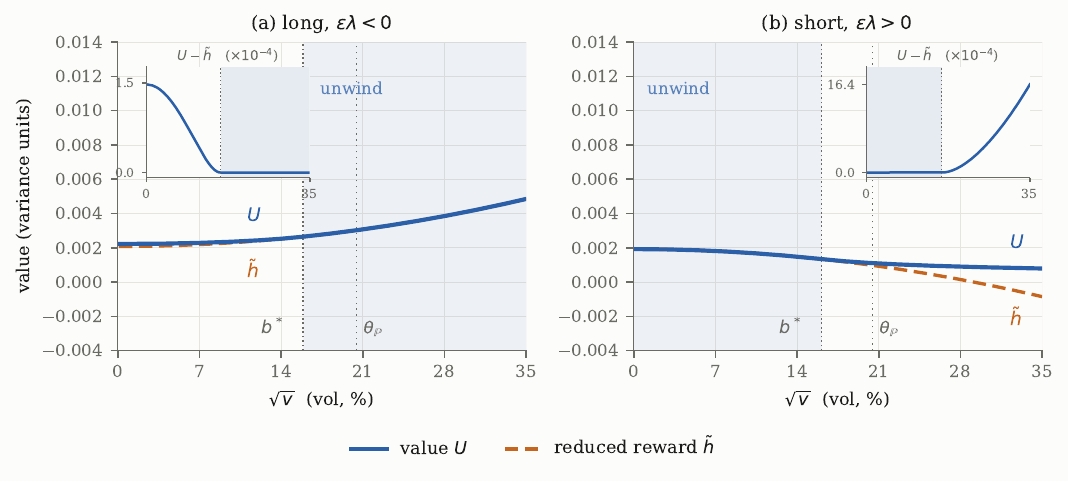}
\caption{The reduced exit problem of Lemma~\ref{lem:reduce} for the perpetual
variance swap at an expected life of two years ($\delta=0.5$) and the base
parameters of
Section~\ref{sec:numerics}. $U$ is the smallest $\delta$-excessive majorant of
the reduced reward $\tilde h$ and leaves it smoothly at $b^*$. Left: long,
$\varepsilon\lambda<0$, $c_m=0$, unwound into a spike. Right: short,
$\varepsilon\lambda>0$, $c_m=0.004$, unwound into a collapse. Both thresholds
correspond to a perpetual fair rate of about $20.3\%$; the dotted line is
$\theta_\PP$. On the main axes the two curves are hard to separate in the left
panel, because there the option value peaks at $1.5\times10^{-4}$ against a
vertical span of $1.8\times10^{-2}$; the insets therefore plot $U-\tilde h$ on
its own scale: strictly positive on the continuation region, identically zero on
the exercise region.}
\label{fig:pair}
\end{figure}

Figure~\ref{fig:pair} shows the two cases side by side. What the insets make
visible is the continuation region, the set on which $U-\tilde h>0$; on the main
axes the option value is too small next to the reward to be seen at all.

\begin{remark}\label{rem:vAfix}
For the fixed-maturity contract of Section~\ref{sec:setup} the analogous
break-even level --- the variance at which the swap marks at zero,
$A_t+\Phi_0(t,v)=0$ in the notation of \eqref{eq:h} --- is
\begin{equation}\label{eq:vAstar}
\begin{aligned}
v^*_{A_t}&=\theta_\QQ-\frac{\kappa_\QQ}{1-e^{-\kappa_\QQ(T-t)}}
\Bigl\{A_t+T(\theta_\QQ-\Sigma^2)-t\theta_\QQ\Bigr\}\\[2pt]
&=\theta_\QQ-\frac{\kappa_\QQ}{1-e^{-\kappa_\QQ(T-t)}}
\Bigl\{(A_t-t\Sigma^2)+(T-t)(\theta_\QQ-\Sigma^2)\Bigr\},
\end{aligned}
\end{equation}
which depends on the path only through the accrued variance $A_t$. It is not the
exercise boundary --- by Proposition~\ref{prop:bound} that is governed by
$v_c(t)$ --- but it is the natural reference point for a position that is under
water.
\end{remark}

\section{When the position is entered}
\label{sec:pair}

A trader who is not yet in the market faces two decisions, and the value of the
exit rule is the reward of the entry problem. From
Remark~\ref{rem:long} onwards the position is the short: entering means selling
variance, going short vega at a fixed $1/\delta$ of duration by
\eqref{eq:pvsvalue}, and unwinding means flattening that exposure --- here by
tear-up rather than by an offsetting trade, as Section~\ref{sec:closed}
describes. Entering at the prevailing fair
rate makes the swap worth zero, so by Lemma~\ref{lem:reduce} and
Proposition~\ref{prop:pcarry} the value of entering at level $v$ is
$R(v;K(v))+U(v)=U(v)-\tilde h(v)-\sbar$, and with an entry concession $\sebar$
the entry problem is
\begin{equation}\label{eq:Jentry}
J(v)=\sup_{\zeta}\ \EE^\PP_v\Bigl[e^{-\delta\zeta}\rho(v_\zeta)\Bigr],
\qquad
\rho:=U-\tilde h-\sbar-\sebar ,
\end{equation}
with $\rho\ge-(\sbar+\sebar)$, equal to the option value of the position net of the
round trip. The bound holds because both $\tau\equiv0$ and $\tau\equiv\infty$ are
admissible in the definition of $U$, the first returning $\tilde h(v)$ and the
second $0$, so that $U\ge\max\{\tilde h,0\}$ pointwise --- equivalently, $U$ is by
Theorem~\ref{thm:dynkin} a non-negative \emph{majorant} of $\tilde h$. Its floor
is attained exactly on the exercise region $\{U=\tilde h\}$, where the option to
wait is worthless and a trader entering there would pay both spreads for
nothing; Proposition~\ref{prop:interval}(ii) turns that observation into the
disjointness of the entry and exercise regions. The factor $e^{-\delta\zeta}$ is again a survival probability, and
what it says is that the contract series itself is offered only until
$T_\delta$: a trader who has not entered by then never does. Section~\ref{sec:closed}
needed the termination event only to be shared between a contract and its
offset; here it must also end the trading opportunity. That is the one
substantive extension the entry problem requires, and it is the reading of
Section~\ref{sec:closed} under which the clock is a lapsing mandate rather than
a contractual term: when the mandate goes, so does the chance to use it. It is
not idle --- without it a flat trader would face no pressure whatever and would
wait indefinitely for a better level. Note what
\eqref{eq:Jentry} does not contain: a running term. That
is the convention $c_0=0$ of Section~\ref{sec:setup} --- while the trader is
flat she accrues no variance and pays no carry, so waiting costs her nothing ---
and we keep it until Section~\ref{sec:idle}, which shows that the entry
threshold, unlike the unwind threshold, is acutely sensitive to it. The same problem can be posed for the fixed-maturity contract of
Section~\ref{sec:setup}. A trader who enters at time $t$ takes a contract struck
fair at that moment, so its mark-to-market vanishes there: by \eqref{eq:h} with
$A_t=0$ this is $\Phi_0(t,v)=0$, and hence $\Phi(t,v)=-s(T-t)$ whichever side she
takes. Her value of holding is therefore
$w(t,v)=u(t,v)+\Phi(t,v)=u(t,v)-s(T-t)$, with $u=w-\Phi$ the excess value of
\eqref{eq:viu}, and the entry spread --- crossed on the same strike, against the
same remaining vega --- costs a further $s_e(T-t)$. The obstacle is thus
\[
\rho(t,v)=u(t,v)-(s+s_e)(T-t),
\]
and $J$ solves $\min\{-\partial_tJ-\LP J,\ J-\rho\}=0$ with $J(T,\cdot)=0$.

One step there deserves spelling out. The strike $\Sigma$ is fixed at inception
for the contract of Section~\ref{sec:setup}, but a trader entering at $t$ takes
a contract struck fair at that moment, so the strike she faces is itself a
function of the state, $\Sigma^\ast(t,v)$. Were $u$ to depend on the strike, the
object in her obstacle would be the composite
$u\bigl(t,v;\Sigma^\ast(t,v)\bigr)$, in which $(t,v)$ enters twice --- directly,
and again through $\Sigma^\ast$ --- and which solves no equation in $(t,v)$; a
separate $u$ would be needed for every entry state. But $u$ carries no such
dependence: the strike enters $\Phi$ only through the additive constant
$-\varepsilon T\Sigma^2$, and both $\partial_t$ and $\LP$ annihilate constants,
so by Proposition~\ref{prop:carry} the forcing term $\mathcal{A}\Phi$ is free of
$\Sigma$. The data of \eqref{eq:viu} are therefore strike-free, its solution $u$
is one and the same function whatever the contract was struck at, and the single
$u$ computed in Section~\ref{sec:finite} serves every entry state.

Both entry problems --- the perpetual one of \eqref{eq:Jentry} and the
fixed-maturity one just posed --- are well posed. Whether either delivers a rule
a trader would use is
decided by the holding cost $c_m$, and \eqref{eq:htD} makes that dependence
explicit: $c_m$ shifts the reduced reward $\tilde h=\varepsilon D-\sbar+c_m/\delta$
by the constant $c_m/\delta$ and leaves its slope alone. Everything in this
section turns on that one fact.

The numerical illustrations in this section use a calibration of their own,
chosen so that the physical law can be held fixed while the premium is varied
--- which the parameters of Section~\ref{sec:numerics}, whose premium is pinned
at $0.60$ volatility points, cannot do. Fixing the shape $\nu$ of \eqref{eq:abz} at $2.84$,
$\kappa_\PP=2.15$ and the median of $v$ at $19\%$ in volatility gives
\begin{equation}\label{eq:cal7}
\theta_\PP=0.0408,\qquad \gamma=0.2485,\qquad \delta=0.5,\qquad \sbar=\sebar=20\ \text{bp},
\end{equation}
and a $99$th percentile of volatility of $34.2\%$.
Throughout, the premium is measured in \emph{volatility points}: the perpetual
fair rate \eqref{eq:pvsrate} and its physical counterpart, both quoted as
volatilities, differ by
$\sqrt{K_\QQ(v)}-\sqrt{K_\PP(v)}$ with
$K_\MM(v)=\theta_\MM+\delta(v-\theta_\MM)/(\delta+\kappa_\MM)$, and this gap is
read at the median of $v$.
Here $\kappa_\QQ$ is reset premium by premium to deliver the stated number of
points, which at three gives
$\kappa_\QQ=1.503$, $\theta_\QQ=0.0583$ and $\lambda=-0.647$ --- a median
strike of $22.97\%$ against a physical $19.98\%$.

\subsection{Three regimes in the holding cost}
\label{sec:regimes}

The first regime is the one in which the exit problem has no content at all.

\begin{proposition}\label{prop:holdforever}
If $\varepsilon D(v)<\sbar-c_m/\delta$ for every $v$, equivalently if
\begin{equation}\label{eq:cmstar}
c_m\ <\ c_m^\ast:=\delta\Bigl(\sbar-\sup_v\varepsilon D(v)\Bigr),
\end{equation}
then $U\equiv0$ and the exercise region is empty: the position, once entered, is
never unwound. The entry reward is then explicit,
\begin{equation}\label{eq:rho0}
\rho(v)=-\varepsilon D(v)-\frac{c_m}{\delta}-\sebar .
\end{equation}
\end{proposition}

\begin{proof}
$\tau\equiv\infty$ is admissible and returns $0$, so $U\ge0$; and
$\tilde h<0$ gives $U\le0$. Hence $U\equiv0$, the contact set
$\{U=\tilde h\}$ is empty, and $\rho=U-\tilde h-\sbar-\sebar=-\tilde h-\sbar-\sebar$, which is
\eqref{eq:rho0}.
\end{proof}

For a short position against a positive premium, $\varepsilon=-1$ and $D>0$, so
\eqref{eq:cmstar} reads $c_m<\delta\,(\sbar+\inf_v D)$: the holding cost is too small
to make the trader give up an accruing premium. Equation \eqref{eq:rho0} then
says the position is worth \emph{owning} wherever that premium exceeds the entry
concession and the expected lifetime carry $c_m/\delta$ --- which is a strictly
larger set than the one on which she opens it, as the next remark records. This is the degeneracy reported in earlier work on the
perpetual contract, and Proposition~\ref{prop:holdforever} localises it: it is a
statement about $c_m$, not about the entry problem. At the parameters of
Section~\ref{sec:numerics} the bound \eqref{eq:cmstar} is reproduced by the
numerical solver to within $0.1$ basis point.\footnote{The supremum in
\eqref{eq:cmstar} is over the whole state space, and for the short it is
$-\inf_vD$, so the bound turns on how small $D$ becomes. Under Heston $D$ is
affine with $D(0)>0$ and it binds where it matters; under a model with
$D(0+)=0$ that infimum is zero, $c_m^\ast$ formally collapses
to $\delta\sbar$, and the exercise region opens at
variance levels the process visits with negligible probability. This is why
Definition~\ref{def:nondeg} carries a reachability tolerance and why the lower
edge plotted in Figure~\ref{fig:window} is the reachable opening rather than the
formal one.}

\begin{remark}[worth holding is not worth opening]\label{rem:ownbuy}
By \eqref{eq:rho0} the position is worth holding exactly on
$\{v:\ -\varepsilon D(v)>c_m/\delta+\sebar\}$, which is the set $\{\rho>0\}$.
That is not the entry region, and the direction of the discrepancy matters.
Since $J\ge0$ always, the contact set $\{J=\rho\}$ on which the trader actually
enters is contained in $\{\rho\ge0\}$; and the containment is strict, because
wherever $\rho$ is positive but small the option of waiting for a richer level is
worth more than taking the position now. The entry region is therefore
\emph{smaller} than the set on which the position is worth holding, not larger.
The gap can be the whole state space: at a carrying charge well below
$c_m^\ast$, $\rho>0$ at every level --- the position is worth holding wherever
the trader happens to find herself --- and she still declines to open it below a
threshold, purely because waiting costs her nothing
(Definition~\ref{def:nondeg} and Section~\ref{sec:idle}).

The exit problem has the same structure, and it is worth saying so once. There
$U\ge\max\{\tilde h,0\}$, so the exercise region $\{U=\tilde h\}$ is likewise a
strict subset of $\{\tilde h\ge0\}$: at the calibration \eqref{eq:cal7} with
$c_m=135$ basis points the reduced reward is positive up to a variance of
$0.030$, while the position is actually unwound only below $0.012$. In both
problems the region is the contact set on which the value meets the obstacle,
and never the set on which the obstacle is merely positive; the gap between them
is the value of waiting.
\end{remark}

\begin{remark}[the long position is excluded]\label{rem:long}
The hypothesis of Proposition~\ref{prop:holdforever} can be met only by a short.
Under the empirical sign of the premium, $\lambda<0$ and $D>0$, and $D$ is
unbounded above: by \eqref{eq:perp} it is affine in $v$ with slope
$-\lambda/\bigl((\delta+\kappa_\QQ)(\delta+\kappa_\PP)\bigr)>0$. For
$\varepsilon=+1$ we then have
$\sup_v\varepsilon D=+\infty$, so $c_m^\ast=-\infty$ and no holding cost
qualifies. The long is degenerate in the opposite way: its exercise region is
non-empty for every $c_m\ge0$, and $\rho<0$ everywhere, falling to its floor
$-(\sbar+\sebar)$ once $c_m$ is a few tens of basis points, so the option to wait is
worthless and the position is never worth opening. That is as it should be, since a long pays
the $\QQ$-rate and accrues the $\PP$-variance and so bleeds the premium: the
only sensible thing to do with one is to close it. Definition~\ref{def:nondeg}
onwards is therefore stated for the short, $\varepsilon=-1$. Two modifications
make the long worth opening --- a negative premium, which merely relabels the
two sides, and a subsidy $c_m<0$, which genuinely reverses the geometry --- and
neither is pursued here.
\end{remark}

The second and third regimes are separated by comparative statics in $c_m$,
and these need nothing beyond the fact that $c_m$ enters $\tilde h$ as a
constant.

\begin{lemma}\label{lem:mono}
Write $\tilde h_m=g+m$ with $g=\varepsilon D-\sbar$ and $m=c_m/\delta$, and let
$U_m$ be the corresponding value, $\mathcal{E}_m=\{U_m=\tilde h_m\}$ its exercise
region --- the contact set on which the position is unwound --- and
$\rho_m=U_m-\tilde h_m-\sbar-\sebar$ the entry reward. Then $m_1\le m_2$ implies
\[
\mathcal{E}_{m_1}\subseteq\mathcal{E}_{m_2},
\qquad
\rho_{m_2}\le\rho_{m_1}\ \text{ pointwise}.
\]
Raising the holding cost therefore enlarges the exercise region and lowers the
entry reward at every level. It does not by itself order the entry
\emph{regions}: see the discussion following
Proposition~\ref{prop:interval}.
\end{lemma}

\begin{proof}
Put $\Delta=m_2-m_1\ge0$. For any stopping time $\tau$,
\[
\EE\bigl[e^{-\delta\tau}(g+m_2)\mathbf 1_{\tau<\infty}\bigr]
=\EE\bigl[e^{-\delta\tau}(g+m_1)\mathbf 1_{\tau<\infty}\bigr]
+\Delta\,\EE\bigl[e^{-\delta\tau}\mathbf 1_{\tau<\infty}\bigr]
\le\EE\bigl[e^{-\delta\tau}(g+m_1)\mathbf 1_{\tau<\infty}\bigr]+\Delta,
\]
because $\delta>0$, which with $\tau\ge0$ gives $0\le e^{-\delta\tau}\le1$ on
$\{\tau<\infty\}$. Taking suprema over $\tau$ gives
$U_{m_2}\le U_{m_1}+\Delta$, and subtracting $g+m_2$ from both sides,
\begin{equation}\label{eq:monokey}
U_{m_2}-(g+m_2)\ \le\ U_{m_1}-(g+m_1)\qquad\text{pointwise}.
\end{equation}
Subtracting the further constant $\sbar+\sebar$ turns \eqref{eq:monokey} into
$\rho_{m_2}\le\rho_{m_1}$, the second claim.

For the first, let $v\in\mathcal{E}_{m_1}$, so that $U_{m_1}(v)=g(v)+m_1$. Then
\eqref{eq:monokey} at $v$ reads $U_{m_2}(v)-(g(v)+m_2)\le0$. The reverse
inequality holds at every $v$ and for every $m$: the stopping time
$\tau\equiv0$ is admissible in the definition of $U_m$ and returns
$\tilde h_m(v)=g(v)+m$, so $U_m\ge\tilde h_m$ --- the majorant property of
Theorem~\ref{thm:dynkin}, already used above in bounding $\rho$ below. Hence $U_{m_2}(v)=g(v)+m_2$, that is $v\in\mathcal{E}_{m_2}$.
\end{proof}

The mechanism is that the carry is paid only while the contract is alive: a
larger $c_m$ adds its full constant $\Delta$ to the reward for stopping now, but
adds to the reward for stopping at $\tau$ only the survival-weighted
$\Delta\,\EE[e^{-\delta\tau}]\le\Delta$, so it tilts every comparison towards
stopping. Neither the diffusion nor the shape of $D$ enters,
so the lemma holds whatever the variance dynamics.

\begin{definition}\label{def:nondeg}
Let $b$ be the unwind threshold, let $d$ be the entry threshold --- the boundary
of the entry region $\{J=\rho\}$ of \eqref{eq:Jentry} --- and let
$F_\infty(x)=\PP(v_\infty\le x)$ be the distribution function of the stationary
law of $v$ under $\PP$. Fix a tolerance $\eta\in(0,\tfrac12)$. The entry--exit
pair is \emph{non-degenerate} if both thresholds exist, $d>b$, and
\begin{equation}\label{eq:reach}
F_\infty(b)>\eta,\qquad F_\infty(d)<1-\eta.
\end{equation}
\end{definition}

Three points about the statement. The entry threshold bounds the contact set
$\{J=\rho\}$ and not the set $\{\rho>0\}$ on which the position is merely worth
holding; Remark~\ref{rem:ownbuy} separates the two. We write $b$ and $d$ without
stars because neither this definition nor Proposition~\ref{prop:interval} uses
the closed form; under Heston they are the optimal $b^*$ and $d^*$ of
Theorems~\ref{thm:exit} and \ref{thm:entry}. And $F_\infty$ is the stationary law, which under Heston is
Gamma with shape $\nu$ and rate $\varsigma$ of \eqref{eq:abz}.

The two inequalities in \eqref{eq:reach} ask that the thresholds be reachable.
Existence and $d>b$ are not enough on their own: a threshold placed far into
either tail describes a rule the physical law triggers so rarely that the pair
is of no practical use, and the footnote to Proposition~\ref{prop:holdforever}
gives a model under which the exercise region opens exactly there. We take
$\eta=0.5\%$ throughout; Appendix~\ref{app:extras} reports what a coarser
tolerance does to the intervals this section produces.

\begin{proposition}\label{prop:interval}
Suppose $D$ is strictly monotone, so that $\beta=\varepsilon D'$ has one sign
throughout, and suppose that sign is negative; suppose further that each region
is one-sided, the exercise region being $\{v\le b\}$ and the entry region
$\{v\ge d\}$ --- the configuration of the short against a positive premium, to
which this section is confined by Remark~\ref{rem:long}. Write $b(c_m)$ and
$d(c_m)$ for the thresholds, and adopt the
conventions $F_\infty(b)=0$ when the exercise region is empty and
$F_\infty(d)=1$ when the entry region is empty. Then
\begin{enumerate}
\item[(i)] $b$ is non-decreasing in $c_m$;
\item[(ii)] the entry and exercise regions are disjoint, so $d>b$ whenever both
are non-empty;
\item[(iii)] if $d$ is also non-decreasing, the set of holding costs for which
the pair is non-degenerate in the sense of Definition~\ref{def:nondeg} is an
interval, possibly empty.
\end{enumerate}
\end{proposition}

\begin{proof}
(i) By Lemma~\ref{lem:mono}, $m_1\le m_2$ gives
$\mathcal{E}_{m_1}\subseteq\mathcal{E}_{m_2}$. Since $\beta<0$ each $\mathcal{E}_m$ is the lower
set $\{v\le b\}$, and one lower set contains another exactly when its threshold
is the larger, so $b$ is non-decreasing in $c_m$.

(ii) On the exercise region $U=\tilde h$, so there
$\rho=U-\tilde h-\sbar-\sebar=-(\sbar+\sebar)<0$. But $J\ge0$ everywhere, because
$\zeta\equiv\infty$ is admissible in \eqref{eq:Jentry} and returns $0$. Hence
$J(v)=\rho(v)$ forces $\rho(v)\ge0$, so no point of the exercise region lies in
the entry region $\{J=\rho\}$: the two are disjoint. As the first is $\{v\le b\}$
and the second $\{v\ge d\}$, disjointness says exactly that $d>b$.

For (iii), $F_\infty\circ b$ is non-decreasing by (i), so
$\{c_m:F_\infty(b)>\eta\}$ is a half-line unbounded above; under the stated
hypothesis $F_\infty\circ d$ is non-decreasing likewise, so
$\{c_m:F_\infty(d)<1-\eta\}$ is a half-line unbounded below; the ordering $d>b$ is automatic by (ii); and the intersection of two such
half-lines is an interval.
\end{proof}

The orientation is what confines the statement. For $\beta>0$ the exercise
region is an upper set, and Lemma~\ref{lem:mono} then lowers its threshold
rather than raising it; in that configuration the roles of the two thresholds
are exchanged and (i)--(iii) hold with the inequalities reversed.

The hypothesis in (iii) is not supplied by Lemma~\ref{lem:mono}. That lemma
lowers the entry obstacle $\rho$ pointwise and hence lowers $J$, but a uniformly
lower obstacle does not by itself order the contact sets $\{J=\rho\}$; the
argument that works for $\mathcal{E}_m$ relies on the perturbation being a constant,
and the decrease in $\rho$ is not. It is supplied instead by the closed form,
once we have it: Proposition~\ref{prop:dmono} below proves $d^*$ strictly
increasing under Heston. We leave (iii) stated as a hypothesis because it is
what the argument uses, and because the argument itself needs nothing of the
dynamics.

\subsection{The entry threshold}
\label{sec:entrythr}

The threshold $d^*$ itself can be pinned down, and by the same argument that
produced $b^*$. Theorem~\ref{thm:exit} works because $\tilde h$ is affine, and
the entry obstacle $\rho=U-\tilde h-\sbar-\sebar$ is not; but by
Proposition~\ref{prop:interval}(ii)
the entry region is disjoint from the exercise region, so it lies where
$U$ is already known in closed form, and there $\rho$ is an explicit
combination of $G$ and an affine function.

\begin{lemma}[the entry obstacle]\label{lem:rho}
Suppose $\beta<0$ and the exercise region is non-empty, with threshold $b^*$
from Theorem~\ref{thm:exit}(ii), and write $k=\sbar+\sebar$ for the round-trip
cost, not to be confused with the discounted carry $m=c_m/\delta$ of
Lemma~\ref{lem:mono}. Then
\begin{equation}\label{eq:rhoclosed}
\rho(v)=
\begin{cases}
-k, & 0<v\le b^*,\\[4pt]
\dfrac{\tilde h(b^*)}{G(b^*)}\,G(v)-\tilde h(v)-k, & v>b^* ,
\end{cases}
\end{equation}
$\rho$ is continuous on $(0,\infty)$ and continuously differentiable at $b^*$,
and
\begin{equation}\label{eq:Lrho}
(\mathcal{L}-\delta)\rho(v)=
\begin{cases}
\delta k, & 0<v<b^* ,\\[4pt]
\delta(\alpha+k)-\beta\kappa_\PP\theta_\PP+\beta(\kappa_\PP+\delta)v,
& v>b^* .
\end{cases}
\end{equation}
The second branch is affine with slope $\beta(\kappa_\PP+\delta)<0$, so
$(\mathcal{L}-\delta)\rho$ is strictly positive on $(0,v_\rho)$ and strictly
negative on $(v_\rho,\infty)$, where
\begin{equation}\label{eq:vrho}
v_\rho=\frac{\beta\kappa_\PP\theta_\PP-\delta(\alpha+k)}
             {\beta(\kappa_\PP+\delta)}
      =v_s^*-\frac{\delta k}{\beta(\kappa_\PP+\delta)}>v_s^*>b^* ,
\end{equation}
with $v_s^*$ as in Lemma~\ref{lem:H}. Consequently
$H_\rho:=(\rho/F)\circ\varphi^{-1}$ is convex on
$(\varphi(0),\varphi(v_\rho)]$ and concave on $[\varphi(v_\rho),0)$.
\end{lemma}

\begin{proof}
See Appendix~\ref{app:proofs}.
\end{proof}

\begin{theorem}[the entry threshold]\label{thm:entry}
Under the hypotheses of Lemma~\ref{lem:rho} the entry region is the upper set
$[d^*,\infty)$, where $d^*>b^*$ is the unique root of
\begin{equation}\label{eq:spd}
\rho'(d^*)\,F(d^*)-\rho(d^*)\,F'(d^*)=0,
\qquad
\rho'(v)=\frac{\tilde h(b^*)}{G(b^*)}\,G'(v)-\beta ,
\end{equation}
and $J(v)=\rho(d^*)F(v)/F(d^*)$ on $(0,d^*)$, with $J=\rho$ on $[d^*,\infty)$ and
$\zeta^*=\inf\{t:v_t\ge d^*\}$ optimal. For
$\beta>0$ the same statements hold with $F$ and $G$ interchanged, the exercise
region the upper set and the entry region the lower one.
\end{theorem}

\begin{proof}
See Appendix~\ref{app:proofs}.
\end{proof}

Two things follow. The one-sidedness that Proposition~\ref{prop:interval} takes
as a hypothesis is, for the perpetual contract under Heston, a conclusion:
\eqref{eq:spd} has a single root and the entry region is the upper set it
bounds. And the entry threshold is as explicit as the exit threshold, both being
unique roots of transcendental equations in the same two confluent
hypergeometric functions. Solving \eqref{eq:spd} by bisection reproduces the
thresholds obtained from the variational inequality to the resolution of its
grid: at the calibration \eqref{eq:cal7} with a three-point premium, $d^*$ is
$0.12161$ from \eqref{eq:spd} against $0.12172$ from the solver at $c_m=200$
basis points, and $0.16745$ against $0.16769$ at $c_m=300$.

Both thresholds are now characterised by smooth-pasting equations, and
differentiating those equations implicitly in $c_m$ gives the comparative
statics. This settles
the monotonicity that Proposition~\ref{prop:interval}(iii) had to assume, and
sharpens part~(i) of the same proposition from a weak inequality to a strict
one with a rate.

\begin{proposition}[the carrying charge moves both thresholds the same way]
\label{prop:dmono}
Under the hypotheses of Theorem~\ref{thm:entry}, $b^*$ and $d^*$ are continuously
differentiable in $c_m$ and both are strictly increasing:
\begin{equation}\label{eq:dbcm}
\frac{\partial b^*}{\partial c_m}
=\frac{-G'(b^*)}{\delta\,\tilde h(b^*)\,G''(b^*)}\;>\;0,
\end{equation}
\begin{equation}\label{eq:ddcm}
\frac{\partial d^*}{\partial c_m}
=\frac{\gamma^2d^*}{2\delta}\cdot
\frac{W(d^*)/G(b^*)-F'(d^*)}{F(d^*)\,(\mathcal{L}-\delta)\rho(d^*)}\;>\;0,
\qquad W:=F'G-FG'>0 .
\end{equation}
For $\beta>0$ both statements hold with $F$ and $G$ interchanged and both
thresholds strictly decreasing. Either way the two move together rather than
apart: the exercise region grows with the charge and the entry region contracts,
so the two reachability requirements of \eqref{eq:reach} bind at opposite ends.
In particular the hypothesis of Proposition~\ref{prop:interval}(iii) holds, and
the set of holding costs admitting a non-degenerate pair is an interval.
\end{proposition}

\begin{proof}
See Appendix~\ref{app:proofs}.
\end{proof}

The exit threshold is the easy half, and worth seeing separately. Write the
smooth-pasting equation \eqref{eq:spG} as $\Psi_b(b,c_m)=\beta G-\tilde hG'$,
writing $\Psi_b$ and $\Psi_d$ for the left-hand sides of the two smooth-pasting
equations \eqref{eq:spG} and \eqref{eq:spd}, whose roots are the two
thresholds. The
charge enters $\tilde h$ as a constant, so $\partial_{c_m}\Psi_b=-G'/\delta>0$;
and in $\partial_b\Psi_b=\beta G'-\tilde h'G'-\tilde hG''$ the first two terms
cancel, because $\tilde h'=\beta$ is exactly the slope that
Theorem~\ref{thm:exit} carries and the linear-payoff case does not, leaving
$-\tilde hG''$. That is \eqref{eq:dbcm}. Nothing is needed beyond $G'<0<G''$
and $\tilde h(b^*)>0$, the latter because $\tilde h(b^*)=0$ would force
$\beta G(b^*)=0$ and hence $\beta=0$.

The entry threshold is the half that resisted, and the step that makes it work
is the derivative of the entry obstacle. By
the envelope theorem the movement of $b^*$ contributes nothing, because $b^*$
solves \eqref{eq:spG}, and what is left is a probability:
\begin{equation}\label{eq:dcmrho}
\frac{\partial\rho}{\partial c_m}(v)
=-\,\EE_v\!\int_0^{T_{b^*}}\!\!e^{-\delta u}\,du
=-\frac1\delta\Bigl(1-\frac{G(v)}{G(b^*)}\Bigr)
\quad\text{for }v>b^* ,
\end{equation}
since $G(v)/G(b^*)=\EE_v[e^{-\delta T_{b^*}}]$ is the survival-weighted
probability of reaching the unwind level from above. The first expression holds
throughout: on the exercise region $T_{b^*}=0$ and $\rho\equiv-\sbar-\sebar$ does
not move with the charge at all, whereas the $G$-ratio is the Laplace transform
only for $v>b^*$ and exceeds $1$ when $v<b^*$. A unit of carrying charge costs the
holder the expected discounted time she will spend carrying the position: this
is Lemma~\ref{lem:mono} with a rate attached, and it explains why that lemma
could not settle the question. The perturbation is not constant in $v$. It is
larger the further $v$ sits above $b^*$, so the obstacle is not merely lowered
but tilted, and a tilt can move a contact set either way.

What resolves it is that the tilt is bounded and the discounting is not. The
magnitude of \eqref{eq:dcmrho} is an expected discounted holding time and so
cannot exceed $1/\delta$,
whereas $F$, which discounts the wait for entry, grows without bound; so the
damage per unit of $F$ rises at first and then falls away, and above the level
at which it turns, raising $c_m$ pushes $d^*$ out rather than in. The content of
the proposition is that $d^*$ always lies above that level, and this follows from the
smooth-pasting relation \eqref{eq:spd} together with the convexity of $G$ --- a
free fact, since $G$ is completely monotone in its argument by Kummer's integral
representation.

What the two derivatives do not settle is the width of the band. Both edges
move out, but \eqref{eq:dbcm} and \eqref{eq:ddcm} are not comparable term by
term, and $d^*-b^*$ is not monotone: it widens at the great majority of the
parameter sets we have solved and narrows at a few, all of them at long contract
lives and large premia. So a heavier carrying charge makes the rule more
selective at both ends, but not reliably slower to complete a round trip.

The three regimes are therefore three pieces of the half-line in $c_m$: hold to
termination below $c_m^\ast$ by Proposition~\ref{prop:holdforever}, a round trip on the
interval of Proposition~\ref{prop:interval}, and never enter above it. Strictly
there is a fourth stretch, between $c_m^\ast$ and the lower edge of the
interval, on which an exercise region exists but lies below the
$\eta$-quantile. The rule there is well defined --- Theorem~\ref{thm:exit}
returns a finite threshold $b$ --- and it is inert. Just above $c_m^\ast$ at the
calibration \eqref{eq:cal7}, $b$ is a variance of $5\times10^{-4}$, an
instantaneous volatility of $2.3\%$, which the stationary law places at its
$0.002$nd percentile: ``unwind if volatility falls to $2.3\%$'' is a perfectly
valid instruction that no equity index will give occasion to follow. The process
is recurrent, so it reaches that level eventually and the rule is not literally
never triggered; the wait is simply of no interest. Where the stretch ends is
fixed by the tolerance $\eta$ rather than by the model, so it is an artefact of
the screen rather than a regime of its own. It is also narrow. At the
calibration \eqref{eq:cal7} under Heston, with $\kappa_\QQ$ set to deliver a
premium of three volatility points, the stretch is the band of carrying charges
running from $c_m^\ast=117$ basis points a year, below which no exercise region
exists at all, up to $122$, above which the threshold has risen past the
$\eta$-quantile and the rule begins to fire.

Figure~\ref{fig:window} locates the interval against the size of the premium.
Its two edges have different origins. The lower edge is Proposition~\ref{prop:holdforever}
together with reachability: below it the trader never unwinds, or unwinds only
at variance levels of no consequence. The upper edge is where $d^*$ leaves the
support: the cost of carry is then so large that the position is worth entering
only at a level the physical law does not produce. Of the two, only the upper
edge depends on the law of $v$ in any essential way: the lower is fixed by
$c_m^\ast$, which involves $D$ and the tear-up concession but not the stationary
distribution, and the law enters it only through the narrow reachability
correction of the fourth stretch just described. The split survives on a
process with a power tail in place of Heston's
exponential one --- the lower edge is essentially unmoved, the upper edge roughly
doubles, and the window widens accordingly, though the entry level itself is
reached no more often.

\begin{figure}[t]\centering
\includegraphics[width=\textwidth]{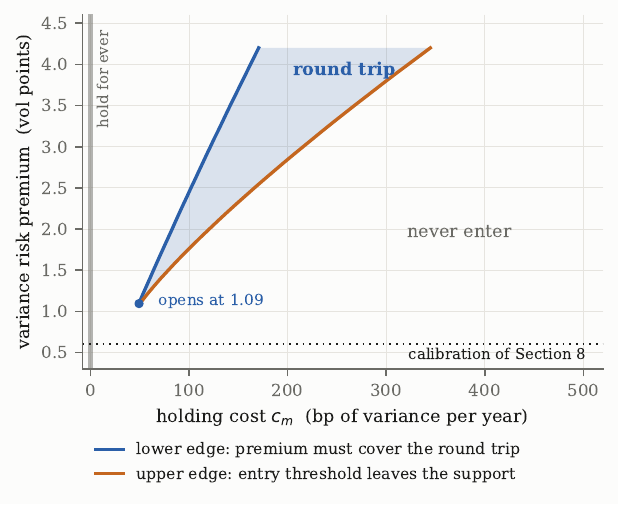}
\caption{Holding costs admitting a non-degenerate entry--exit pair
(Definition~\ref{def:nondeg}) for a short perpetual position, against the size
of the variance risk premium. The median of the stationary law of $v$ is fixed
at $19\%$ in volatility, with $\sbar=\sebar=20$ basis points and $\delta=0.5$. Below a
premium of $1.09$ volatility points, no holding cost admits a round trip.}
\label{fig:window}
\end{figure}

One parameter is left to check. Everything above fixes $\delta=0.5$, a contract
whose expected life is two years, and a longer-lived contract weights distant
states more heavily, which might be expected to pull the entry threshold down.
Table~\ref{tab:horizon} sweeps it.

\begin{table}[t]\centering
\begin{tabular}{rrrrr}
\toprule
expected life $1/\delta$ & $c_m$ (bp) & exit pctile & entry pctile & entry / exit rate\\
\midrule
$\phantom{0}1$y & $135$ & $\phantom{0}6.6$ & $94.7$ & $27.03\%\,/\,20.47\%$\\
$\phantom{0}2$y & $135$ & $\phantom{0}6.9$ & $95.4$ & $25.66\%\,/\,21.63\%$\\
$\phantom{0}4$y & $135$ & $\phantom{0}6.3$ & $96.6$ & $24.74\%\,/\,22.33\%$\\
$\phantom{0}8$y & $135$ & $\phantom{0}5.5$ & $97.8$ & $24.14\%\,/\,22.73\%$\\
$16$y & $145$ & $13.9$ & $98.9$ & $23.79\%\,/\,23.00\%$\\
$32$y & $145$ & $13.8$ & $99.4$ & $23.54\%\,/\,23.10\%$\\
\bottomrule
\end{tabular}
\caption{Lowest entry percentile attainable at each contract life, at the
calibration of \eqref{eq:cal7}, short position, over holding costs $c_m$ subject
to the exit threshold remaining above the fifth percentile --- a screen distinct
from, and weaker than, the tolerance $\eta$ of Definition~\ref{def:nondeg},
chosen so that a threshold exists to report at every life. Rates are the
perpetual fair rates at the two thresholds. Each row is recalibrated: holding
$\kappa_\PP$, the shape $\nu$ and the median fixed, $\kappa_\QQ$ is reset so
that the premium is three volatility points at that $\delta$.}
\label{tab:horizon}
\end{table}

Length does not help; it hurts, and monotonically. The entry percentile is
lowest at the shortest life, $94.7$ at one year, and rises steadily to $99.4$ at
thirty-two. The last column says why, and the reason is a property of the
instrument rather than of its holder. As $\delta$ falls,
$\partial K/\partial v=\delta/(\delta+\kappa_\QQ)$ falls with it, the contract's
rate loses its sensitivity to the state, and both thresholds converge on
$\sqrt{\theta_\QQ}$: at one year they are $27.03\%$ and $20.47\%$, a band of six
and a half volatility points, and at thirty-two years $23.54\%$ and $23.10\%$, a
band of less than half a point. A long-lived contract is quoted at a rate that
barely moves with the state, so there is nothing left to time. The shortest-lived
contract has the widest band and the most reachable entry, and even there the
short is not put on until the perpetual rate reaches the $95$th percentile. This is a statement about how to write the
contract, not about how patient to be.

\subsection{The cost of standing aside}
\label{sec:idle}

This subsection stays with the short of Section~\ref{sec:regimes}, for which
entering means selling the swap into an elevated rate.

Nothing so far makes waiting cost the flat trader anything in cash. While she
holds no position she accrues no variance and pays no carrying charge, and by
Proposition~\ref{prop:cancel} and Remark~\ref{rem:whatdelta} no discount factor
erodes her eventual profit either, so \eqref{eq:Jentry} carries no running term.
The only penalty for delay is the survival factor $e^{-\delta\zeta}$, which
withdraws the opportunity at a constant rate but takes nothing from her while it
lasts, and at $\delta=0.5$ it is far too weak to compete with the value of a
better level. That is why
the entry threshold stands so far above the level at which the position first
becomes worth holding: at the calibration of Section~\ref{sec:working} $\rho>0$
from the $76$th percentile upwards, but the optimal entry threshold sits at the
$94.8$th, and the whole of that gap is the value of holding out for a better
level at no cost.

Suppose instead that the opportunity cost $c_0$ admitted in
Section~\ref{sec:setup}, and set to zero there, is allowed to bite. Then
\[
J(v)=\sup_\zeta\EE^\PP_v\Bigl[\int_0^\zeta e^{-\delta u}(-c_0)\,du
+e^{-\delta\zeta}\rho(v_\zeta)\Bigr]
=-\frac{c_0}{\delta}
+\sup_\zeta\EE^\PP_v\Bigl[e^{-\delta\zeta}\Bigl(\rho+\frac{c_0}{\delta}\Bigr)(v_\zeta)\Bigr],
\]
which is the entry problem with its obstacle raised by the constant $c_0/\delta$.
Lemma~\ref{lem:mono} used nothing about the reward except that the perturbation
is a constant, so it applies verbatim, with the sign reversed: the entry region
grows with $c_0$ and the entry threshold falls.

\begin{table}[t]\centering
\begin{tabular}{rrr}
\toprule
$c_0$ (bp per year) & entry rate & entry percentile\\
\midrule
$\phantom{00}0$ & $22.62\%$ & $94.8$\\
$\phantom{0}25$ & $21.17\%$ & $83.1$\\
$\phantom{0}50$ & $19.81\%$ & $58.0$\\
$\ge100$ & $17.32\%$ & $\phantom{0}0.0$\\
\bottomrule
\end{tabular}
\caption{Entry threshold for the \emph{short} against the opportunity cost of
idle capital, at the calibration of Section~\ref{sec:working} with $c_m=120$
basis points, for which the unwind threshold is a perpetual rate of $18.20\%$ at
the $12.4$th percentile. Entering here means selling the swap: the threshold is
the rate at or above which the short is put on. Rates are perpetual fair rates; percentiles are of the stationary
law of $v$ under $\PP$. At $c_0\ge100$ basis points the trader enters at once
whatever the level.}
\label{tab:idle}
\end{table}

The sensitivity is severe. Fifty basis points a year --- two fifths of the
holding charge at this calibration, and small beside any hurdle rate a desk
would actually be set --- moves the entry threshold from the $95$th percentile of
the stationary law to the $58$th, and past about a hundred basis points the trader
enters immediately at any level. This is the sharper form of the asymmetry. The
exit rule of Section~\ref{sec:closed} needs no such quantity: it follows from the
premium and the trading spread, both of which are observable. The entry rule
needs one, and it is a property of the trader rather than of the market, which
is why the two decisions are not mirror images however the problem is posed.

\begin{remark}[the continuation region is unbounded above]
\label{rem:nostop}
For a short against a positive premium the continuation region is
$(b^*,\infty)$: the rule never closes the position into a rising market. This is
a consequence of the objective rather than an accident. $D$ is increasing, so
the premium accrues faster the higher variance goes, and \eqref{eq:pvsobj}
rewards the trader for staying short through an arbitrarily large spike. The
objective contains no capital constraint, no margin call and no possibility of
forced liquidation, and it lets the trader carry the position at
$c_m$ per unit time whatever the mark. Read as a statement about where the
premium is earned, the conclusion is correct. Read as a trading rule it is the
familiar short-volatility payoff with no stop on the losing side, and a desk
would impose a stop-loss: an upper barrier on $v$ at which the position is cut
whatever the premium still on offer. Imposing one is not innocuous. Cutting the
position at the barrier gives up the paths on which the premium accrues fastest,
so continuing is worth less at every level beneath it, and $b^*$ rises in
consequence: the stop-loss changes where she unwinds at the bottom, not only at
the top. We have not pursued this, but of the modifications
one might make it is the most likely to matter in practice.
\end{remark}
Whichever side is being traded, the asymmetry between the two decisions is real,
and more robust than a single calibration would suggest. An exit rule follows from the premium and
the cost of trading alone. An entry rule exists, and Proposition~\ref{prop:interval}
locates exactly the holding costs that admit one, but across the holding costs
that admit a round trip at all, and across expected contract lives from one year
to thirty-two, the entry threshold never falls below roughly the $95$th percentile
of the physical law, and not below the $87$th even when
the affine link is abandoned and the premium is generated wholly on the pricing
side. It is not an artefact of the horizon built into the contract, since a
longer-lived contract is strictly worse, and by Appendix~\ref{app:extras} not an
artefact of the premium at all, since a state-dependent carrying charge that
reproduces the whole structure with $\lambda=0$ leaves the entry threshold
exactly where it was. It is, however, an artefact of waiting that costs nothing
in cash, and Section~\ref{sec:idle} isolates the single parameter that removes
it. The
calibration of Section~\ref{sec:numerics} carries a premium of $0.60$
volatility points, below the $1.09$ at which the Heston window opens, so its
entry threshold sits above the $99$th percentile while its exercise threshold
sits at the $30$th; but the qualitative conclusion does not depend on that
choice. Timing the exit of a variance swap is a problem the premium solves.
Timing the entry is one it does not solve alone.

\section{Numerical illustration}
\label{sec:numerics}

\subsection{Parameters and units}
\label{sec:units}

We take $\kappa_\QQ=2.0$, $\theta_\QQ=0.045$ (a strike of $21.2\%$ in volatility
terms), $\gamma=0.26$ and $\lambda=-0.15$, so that $\kappa_\PP=2.15$ and
$\theta_\PP=0.0419$ ($20.5\%$). The shape of \eqref{eq:abz} is $\nu=2.66$, so the Feller condition holds under
both measures --- the affine link makes $\nu$ the same under each --- and the
stationary law of $v$ under $\PP$ is $\mathrm{Gamma}$ with shape $\nu$ and rate
$63.6$. The unwind spread is $s=0.002$ in variance units per year of remaining
maturity, about half a volatility point at these levels, and the entry spread is
the same. For the perpetual we take the tear-up concessions
$\sbar=\sebar=0.002$, so that closing a perpetual costs what closing a one-year
dated swap costs; Remark~\ref{rem:spreadconv} reports what happens if the two
are instead tied together at $\sbar=s/\delta$. Where a carrying charge is needed
we use $c_m=0.004$ per year. For the fixed-maturity contract we set $T=1$ and
$\Sigma^2$ equal to the fair strike at $v_0=\theta_\PP$, namely
$\Sigma^2=0.0436$ ($20.9\%$).

Thresholds are naturally computed in the state variable $v$, but no one quotes
instantaneous variance. What a trader observes is the fair strike for the
residual maturity,
\begin{equation}\label{eq:xi}
\xi(t,v)=\theta_\QQ+(v-\theta_\QQ)\,\frac{\gQ(t)}{T-t},
\end{equation}
which is a duration-weighted average of expected future variance and is
therefore far less dispersed than $v$ itself: at $t=0$ with $T=1$ and
$\kappa_\QQ=2$ the sensitivity $\gQ(0)/T$ is only $0.43$. Reporting a rule in
$\sqrt{v}$ rather than in $\sqrt{\xi}$ inflates every band by a factor of
roughly two and makes thresholds look implausible --- a level of $\sqrt v=40\%$
is a one-year strike of $31\%$, and $\sqrt v=7\%$ is a strike of $17\%$. We
therefore quote all thresholds as $\sqrt{\xi}$ for the dated contract and as the
perpetual fair rate $\sqrt K$ of \eqref{eq:pvsrate} for the perpetual, in
volatility per cent, each with the percentile of the threshold under the
stationary law of $v$ under $\PP$. That percentile answers the question a rule
must answer before any other: does it ever fire?

The diffusion coefficient is what makes those percentiles meaningful, since
the shape $\nu$ of \eqref{eq:abz} is the shape parameter of the
stationary law, and it must be comfortably above one. At $\gamma=0.40$ it is
$1.12$, the stationary density is nearly exponential, and the model assigns some $8\%$
probability to instantaneous volatility below $7\%$ --- a level no equity index
has visited. At $\gamma=0.26$ we have $\nu=2.66$ and the implied one-year strike
has percentiles
\[
1\%:16.7\qquad 25\%:18.8\qquad 50\%:20.4\qquad 75\%:22.2\qquad 99\%:28.1
\]
in volatility terms, which is a recognisable description of a large-cap index.

Finally, a word on the units of $\delta$, since it is easily misread as a
funding spread. By Remark~\ref{rem:whatdelta} it is neither a funding spread nor
a discount rate but the termination rate written into the contract, so
$1/\delta$ is a tenor and belongs on the same axis as the maturity of a dated
swap. A value of $\delta=0.05$ is therefore not a five per cent cost of carry
but a contract with an expected life of twenty years, which is why it produces
such selective rules: by \eqref{eq:pvsrate} its quoted rate carries almost no
information about the current state. Funding is carried by $c_m$ throughout.

\subsection{The perpetual contract}
\label{sec:working}

Two quantities organise the thresholds of Theorems~\ref{thm:exit}
and~\ref{thm:entry}. The first is the carrying charge. For a short position the
exercise region is empty unless $c_m$ exceeds the threshold $c_m^\ast$ of
Proposition~\ref{prop:holdforever}, which at these parameters is
$\delta(\sbar+\inf_vD)=0.00304$ per year, reproduced by the numerical solver to
within $0.1$ basis point: a perpetual with positive carry and no expiry is a
money machine, and there is no reason ever to close it. Once $c_m$ passes that
threshold the exit level rises quickly, from a perpetual rate of $19.7\%$ at
$c_m=0.0035$ to $23.6\%$ at $c_m=0.008$, at an expected life of two years. The
long position needs no carrying charge at all: it bleeds the premium, and at
$c_m=0$ and $\delta=0.5$ it is closed once the rate reaches $20.3\%$, the $29$th
percentile.

The second is $\lambda$, which enters the slope of the reduced reward directly
through \eqref{eq:pbeta}. As $|\lambda|$ rises from $0.05$ to $0.15$ the short's
exit level falls from $24.9\%$ to $20.3\%$ and its percentile from $99$ to $30$;
beyond $|\lambda|=0.30$ the exercise region empties, because a premium that
large outruns any plausible carrying charge. This is
Corollary~\ref{cor:plam} in numbers: the premium is what makes the timing
problem exist, and enough of it makes the position one to hold rather than to
time.

The base parameters carry a premium of only $0.60$ volatility points and
therefore lie below the opening of the window of Figure~\ref{fig:window}; they
exhibit the third regime of Section~\ref{sec:pair}, in which an unwind rule
exists and no entry rule does. One calibration inside the window is worth
recording. Keeping the perpetual strike near $21\%$ but slowing the pricing
reversion, $\kappa_\QQ=1.0$, $\theta_\QQ=0.045$, $\gamma=0.20$ and
$\lambda=-0.6$ give $\kappa_\PP=1.60$ and $\theta_\PP=0.02812$, a physical
long-run volatility of $16.8\%$, $\nu=2.25$, and a premium of
$3.07$ volatility points. The window is then $c_m\in[100,210]$ basis points, and
at $c_m=120$ the short is entered when the perpetual rate reaches $22.62\%$, the
$94.8$th percentile of the stationary law, and unwound when it falls to
$18.20\%$, the $12.4$th. Simulation under $\PP$ puts the mean time from
stationarity to entry at $4.0$ years and from entry to unwind at $2.3$.

Slowing $\kappa_\QQ$ is what makes this work, and it is the only lever that
does: holding $\kappa_\QQ=2$ and sweeping $\lambda$ and $\gamma$ over their
admissible ranges never brings the entry threshold below the $97$th percentile,
because the premium enters through $\kappa_\PP=\kappa_\QQ-\lambda$ and a fast
pricing reversion forces a fast physical one, which holds $v$ near
$\theta_\PP$. Even in the favourable calibration the entry level is at the $95$th
percentile, which is the point of Table~\ref{tab:horizon}: this is close to the
best the model does, not a lucky corner of it.

\subsection{The dated contract}
\label{sec:finite}

The perpetual is what admits a closed form; the dated contract of
Section~\ref{sec:setup} does not, and we solve it numerically to check that the
geometry the reduction predicts is the geometry the accruing contract actually
has. By Propositions~\ref{prop:reduce} and \ref{prop:carry} the excess value
$u:=w-\Phi\ge0$ satisfies
\begin{equation}\label{eq:vi}
\min\Bigl\{-\partial_tu-\LP u-\bigl(\varepsilon\lambda \gQ(t)v+s-c_m\bigr),\ u\Bigr\}=0
\quad\text{on }(0,T)\times(0,\infty),
\qquad u(T,\cdot)=0,
\end{equation}
whose free boundary $b_T(t)=\partial\{u(t,\cdot)=0\}$ is the curve introduced in
Section~\ref{sec:carry}, the dated counterpart of the perpetual threshold $b^*$
of Theorem~\ref{thm:exit}. We solve \eqref{eq:vi} by a
fully implicit finite-difference scheme with upwinded drift and Howard iteration
on the exercise set, treating $v=0$ and the truncation boundary as degenerate
transport rows. The scheme is monotone and the solution is independent of the
truncation level; $b_T(0)$ is stable to four significant figures under refinement
from $(N_v,N_t)=(500,400)$ to $(4000,3200)$. Figure~\ref{fig:bdy} reports the
boundary for a long position with $\varepsilon\lambda<0$ --- the first row of
Table~\ref{tab:cases}, in which the motive for delay is the decaying unwind
spread --- and for a short position with $\varepsilon\lambda>0$.

\begin{figure}[t]\centering
\includegraphics[width=\textwidth]{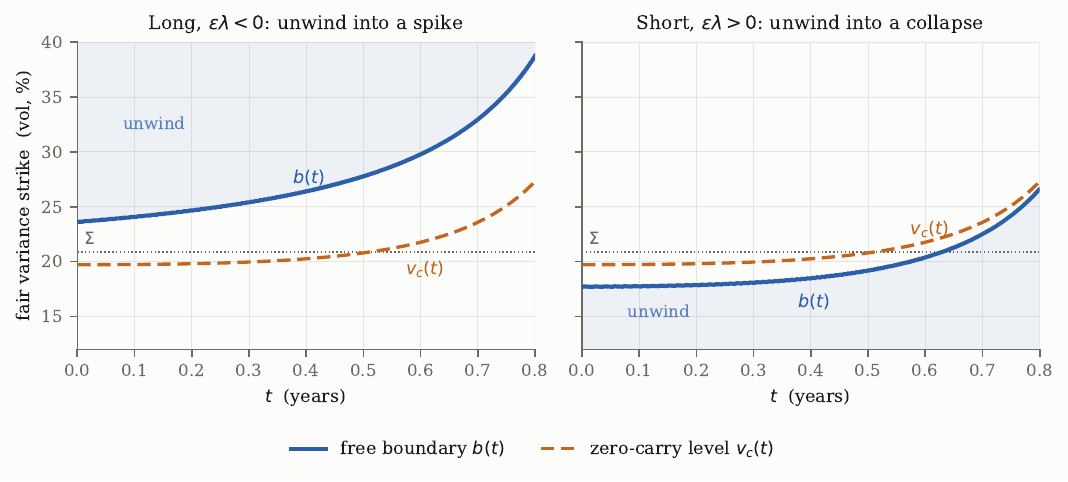}
\caption{Free boundary $b_T(t)$ of \eqref{eq:vi} (solid) and myopic level $v_c(t)$
of \eqref{eq:vc} (dashed), quoted as the fair variance strike \eqref{eq:xi}.
Shading marks the exercise region. Left: long, $\varepsilon\lambda<0$, $c_m=0$.
Right: short, $\varepsilon\lambda>0$, $c_m=0.004$. The dotted line is the
strike $\Sigma$ at which the swap was struck.}
\label{fig:bdy}
\end{figure}

Three features carry over, and one does not. The boundary is far from constant:
over three quarters of the life of a one-year swap it moves by about $1200$
basis points of volatility for the long position and $650$ for the short, so any
treatment that freezes the accrued variance and the residual maturity, and so
produces a single threshold, will misstate the rule by far more than the
transaction costs it is trying to economise on. Both rules fire with a plausible
frequency --- at inception the long unwinds at the $87$th percentile and the
short at the $9$th --- and both sit at the same order as the perpetual exit
thresholds above, so it is not the maturity that makes an exit rule usable. And
the premium sets the urgency in the same way: as $|\lambda|$ falls from $0.50$ to
$0.05$ the long trigger rises from $17.8\%$ to $36.5\%$ and the option value of
waiting grows from $0.6$ to $10.9$ volatility points, while for the short the
exercise region empties once $|\lambda|\ge0.30$.

What does not carry over is the behaviour near maturity, and the reason is
instructive. For the short the boundary rises to meet the myopic level,
$b_T(t)\uparrow v_c(t)$ as $t\uparrow T$: the option to wait is worth nothing when
there is no time left in which to wait. For the long both levels diverge like
$1/\gQ(t)$, because the negative carry $\lambda \gQ(t)v$ that motivates the exit
vanishes with the remaining duration while the spread saving $s(T-t)$ does not
vanish as fast; the trigger runs away and the position is held to expiry. The
perpetual has no maturity and so none of this, which is part of why it is the
tractable object. But maturity is not only a nuisance: it supplies the decaying
unwind spread, and for the long against the premium --- the first row of
Table~\ref{tab:cases}, whose carry is negative at every level --- that saving is
the only motive for delay. Without it $\mathcal{A}\Phi<0$ everywhere, and the
position is unwound at once from any state.

\section{Conclusion}
\label{sec:conclusion}

The entry and exit timing of a perpetual variance swap has a closed-form
solution, both thresholds being unique roots of smooth-pasting equations in the
confluent hypergeometric solutions of the CIR generator
(Theorems~\ref{thm:exit} and \ref{thm:entry}). Getting there was a reduction
rather than a new piece of machinery. Under $\QQ$ the problem has no content;
under $\PP$ the accrued variance separates exactly, the maturity, strike and
costs collapse into the single term $\varepsilon\lambda \gQ(t)v+(s-c_m)$, and on
the perpetual the running reward comes off and the entry strike cancels, leaving
a reward affine in the state. An affine reward on a CIR process is a problem
\citet{dayanik2003} and \citet{leung2014} have solved; the contribution is that
a variance swap should have been expected to produce one.

Opening the position and closing it are not mirror images. For the short --- the only side worth
opening under the empirical sign of the premium --- the holding cost is what
makes an unwind rule exist at all, and the charges admitting both a reachable unwind and
a reachable entry form an interval; inside it the entry threshold still sits
near the $95$th percentile of the physical law, because a trader who is out of
the market pays nothing to stay out. Charging idle capital fifty basis points a
year moves it to the $58$th. The premium and the trading spread, both
observable, say when to close a variance swap; saying when to open one needs a
number that describes the trader rather than the market.

One limitation is worth naming. \citet{dewbecker2017} locate the price of
variance risk at one and two months and find it indistinguishable from zero
beyond a quarter, which is not where a perpetual contract puts its weight;
Appendix~\ref{app:extras} shows the construction survives without a premium,
driven by a state-dependent carrying charge instead, at the cost of
conditioning the result on the shape of a trader's book rather than on market
prices.

\section*{Declaration of generative AI use}

The author used Anthropic's Claude during this work, both for the research ---
proof development, symbolic verification, and the code behind the numerical
results --- and for drafting and editing the manuscript. The author checked all
arguments and results and takes full responsibility for the content of this
article.

\appendix

\section{Proofs}
\label{app:proofs}

\begin{proof}[Proof of Lemma~\ref{lem:H}]
Continuity and differentiability are inherited from those of $F$ and $\varphi$.
Differentiating $H(\varphi(y))=\tilde h(y)/F(y)$ once gives the first expression
in \eqref{eq:Hprime}; differentiating again and using that $F$ solves
$\mathcal{L}F=\delta F$ together with $\varphi'=\mathrm{Wr}/F^2$, where
$\mathrm{Wr}$ is the Wronskian of $F$ and $G$, gives the second. Since
$\tilde h'\equiv\beta$ by Proposition~\ref{prop:pcarry}, the numerator of $H'$
is $\beta F-\tilde hF'$.

For $H(0)$: $\varphi(y)\to0$ as $y\to\infty$, and by \eqref{eq:derivs} $F'$ is
strictly increasing and unbounded, so by L'H\^opital
$H(0)=\lim_{y\to\infty}\tilde h(y)^+/F(y)=\lim_{y\to\infty}\beta/F'(y)=0$. The
sign statement is immediate from $F>0$, the monotonicity of $\varphi$, and the
fact that $v_0^*$ is the unique zero of the affine $\tilde h$.

For convexity, $\tilde h''=0$ and $\tilde h'=\beta$ give
$(\mathcal{L}-\delta)\tilde h(v)=\kappa_\PP(\theta_\PP-v)\beta-\delta\tilde h(v)$,
which is affine and strictly decreasing in $v$ with coefficient
$-\beta(\kappa_\PP+\delta)<0$; it therefore has a unique zero $v_s^*$ and is
positive to its left. By the second expression in \eqref{eq:Hprime}, $H''$ has
the same sign, and $\varphi$ is increasing, which gives the stated convexity and
concavity.
\end{proof}

\begin{proof}[Proof of Theorem~\ref{thm:exit}]
We give (i); (ii) is the same argument under Proposition~\ref{prop:dk} with the
roles of the endpoints exchanged and $G$ in place of $F$. By Lemma~\ref{lem:H},
$H(\varphi(v_0^*))=0$ and $H(0)=0$ with $H>0$ in between, so by Rolle's theorem
$H'$ vanishes at some interior $z^*=\varphi(b^*)$; by the concavity of $H$ on
$[\varphi(v_s^*),0)$ established in Lemma~\ref{lem:H}, and since $H$ is
increasing to the left of $\varphi(v_0^*)\vee\varphi(v_s^*)$, that point is
unique and is a maximum. From the first expression in \eqref{eq:Hprime}, $z^*$
solves \eqref{eq:spF}. The smallest decreasing concave majorant of $H$ is then
\[
\widehat H(z)=\begin{cases}H(z^*), & z<z^*,\\ H(z), & z\ge z^*,\end{cases}
\]
and Proposition~\ref{prop:dk2} gives $U(v)=F(v)\widehat H(\varphi(v))$, which is
the stated form; adding $R$ gives $V$ by Lemma~\ref{lem:reduce}. Optimality of
$\tau^*=\inf\{t:v_t\ge b^*\}$ follows from Proposition~\ref{prop:dk}. If $H$ has
no interior stationary point the exercise region is empty and the position is
never unwound.
\end{proof}

\begin{proof}[Proof of Lemma~\ref{lem:rho}]
By Theorem~\ref{thm:exit}(ii) the exercise region is $(0,b^*]$, on which
$U=\tilde h$, and the continuation region is $(b^*,\infty)$, on which
$U=\tilde h(b^*)G/G(b^*)$. Substituting into $\rho=U-\tilde h-k$ gives
\eqref{eq:rhoclosed}; continuity and the $C^1$ fit at $b^*$ are the
value-matching and smooth-pasting conditions of that theorem.

For \eqref{eq:Lrho}: on $(0,b^*)$, $\rho$ is the constant $-k$, so
$(\mathcal{L}-\delta)\rho=-\delta(-k)=\delta k$. On $(b^*,\infty)$,
$\mathcal{L}G=\delta G$ gives $(\mathcal{L}-\delta)G=0$, so
$(\mathcal{L}-\delta)\rho=-(\mathcal{L}-\delta)\tilde h+\delta k$; with
$\tilde h=\alpha+\beta v$, $\tilde h''=0$ and $\tilde h'=\beta$,
\[
(\mathcal{L}-\delta)\tilde h(v)=\beta\kappa_\PP(\theta_\PP-v)-\delta(\alpha+\beta v)
=\beta\kappa_\PP\theta_\PP-\delta\alpha-\beta(\kappa_\PP+\delta)v ,
\]
which on substitution yields the second branch. Its slope is
$\beta(\kappa_\PP+\delta)<0$ because $\beta<0$, so it is strictly decreasing
and vanishes only at $v_\rho$; the first branch is the positive constant
$\delta k$.

It remains to place $v_\rho$. Subtracting the two displays gives
$v_\rho=v_s^*-\delta k/(\beta(\kappa_\PP+\delta))$, and $\beta<0$ makes the
second term positive, so $v_\rho>v_s^*$. For $v_s^*>b^*$, note that where
$(\mathcal{L}-\delta)\tilde h>0$ the obstacle is a strict subsolution and
stopping is not optimal, so the exercise region is contained in
$\{(\mathcal{L}-\delta)\tilde h\le0\}$; for $\beta<0$ that function is
\emph{increasing} in $v$ with zero at $v_s^*$, so the containment reads
$(0,b^*]\subseteq(0,v_s^*]$. Hence $(\mathcal{L}-\delta)\rho>0$ on all of
$(0,v_\rho)$, including across $b^*$, and $<0$ beyond. The statement about
$H_\rho$ then follows from the second expression in \eqref{eq:Hprime}, whose
derivation used only that the obstacle is $C^2$ on the relevant interval, with
$\rho$ in place of $\tilde h$: $H_\rho''$ carries the sign of
$(\mathcal{L}-\delta)\rho$, and $\varphi$ is increasing.
\end{proof}

\begin{proof}[Proof of Theorem~\ref{thm:entry}]
The entry problem $J(v)=\sup_\zeta\EE^\PP_v[e^{-\delta\zeta}\rho(v_\zeta)]$ is
\eqref{eq:general} with obstacle $\rho$, which is continuous by
Lemma~\ref{lem:rho}, so Propositions~\ref{prop:dk} and \ref{prop:dk2} apply.
By Lemma~\ref{lem:rho}, $H_\rho$ is convex then concave with the switch at
$\varphi(v_\rho)$, which is the structure Lemma~\ref{lem:H} supplies for
$\tilde h$ when $\beta>0$. Although we are in $\beta<0$, it is therefore
branch~(i) of Theorem~\ref{thm:exit} whose argument transfers, with $\rho$ for
$\tilde h$: $H_\rho$ has a unique
interior stationary point $\varphi(d^*)$, which is a maximum and solves
$\rho'F-\rho F'=0$ by the first expression in \eqref{eq:Hprime}; the smallest
decreasing concave majorant is constant to its left and equal to $H_\rho$ to its
right; and $J(v)=F(v)\widehat H_\rho(\varphi(v))$ gives the stated form, with
$\zeta^*=\inf\{t:v_t\ge d^*\}$ optimal by Proposition~\ref{prop:dk}.

That $d^*>b^*$ is Proposition~\ref{prop:interval}(ii): $\rho=-k<0$ on the
exercise region while $J\ge0$ always, so no point of $(0,b^*]$ lies in
$\{J=\rho\}$. The case $\beta>0$ is the same argument with the roles of the two
endpoints exchanged, $G$ in place of $F$, and Proposition~\ref{prop:dk} in place
of Proposition~\ref{prop:dk2}.
\end{proof}

\begin{proof}[Proof of Proposition~\ref{prop:dmono}]
Throughout $\beta<0$, $k:=\sbar+\sebar>0$ and $m:=c_m/\delta$, so that
$\tilde h(v;m)=\tilde h(v;0)+m$ and $\partial_m\tilde h\equiv1$. We use
repeatedly that $G(v)=\mathcal{U}(\mathsf{a},\nu;\varsigma v)$ of \eqref{eq:FG},
with $\mathsf{a}=\delta/\kappa_\PP>0$, is completely monotone in its argument, by
Kummer's integral representation
\[
\mathcal{U}(p,q;x)=\frac{1}{\Gamma(p)}\int_0^\infty e^{-xt}t^{p-1}(1+t)^{q-p-1}\,dt ,
\qquad p>0,\ x>0,
\]
which exhibits $\mathcal{U}$ as the Laplace transform of a positive measure;
hence $G>0$, $G'<0$ and $G''>0$. Likewise $F(v)=M(\mathsf{a},\nu;\varsigma v)$ has
non-negative Taylor coefficients because $\mathsf{a},\nu>0$, so $F,F',F''>0$.

\emph{The exit threshold.} Write \eqref{eq:spG} as
$\Psi_b(b,m):=\beta G(b)-\tilde h(b;m)G'(b)$. Then $\partial_m\Psi_b=-G'(b)>0$, and
since $\tilde h'\equiv\beta$,
\[
\partial_b\Psi_b=\beta G'(b)-\tilde h'(b)G'(b)-\tilde h(b)G''(b)=-\tilde h(b)G''(b).
\]
At $b=b^*$ we have $\tilde h(b^*)>0$: the exercise region is non-empty and $U\ge0$
gives $\tilde h(b^*)\ge0$, while $\tilde h(b^*)=0$ would reduce \eqref{eq:spG} to
$\beta G(b^*)=0$ and hence $\beta=0$. So $\partial_b\Psi_b(b^*)<0$, the implicit
function theorem applies, and
\[
\frac{\partial b^*}{\partial m}=-\frac{\partial_m\Psi_b}{\partial_b\Psi_b}
=\frac{-G'(b^*)}{\tilde h(b^*)G''(b^*)}>0 ,
\]
which is \eqref{eq:dbcm} after dividing by $\delta$. This strengthens
Proposition~\ref{prop:interval}(i), whose contact-set argument gives only a weak
inequality. For $\beta>0$ replace $G$ by $F$ throughout: $\partial_m\Psi_b=-F'<0$,
$\partial_b\Psi_b=-\tilde hF''<0$, and $\partial b^*/\partial m<0$, so the
exercise region $[b^*,\infty)$ again enlarges.

\emph{Step 1: the derivative of the entry obstacle.} On the interval
$(b^*,\infty)$, Theorem~\ref{thm:exit}(ii) gives
$U(v;m)=\Lambda(b^*(m),m)\,G(v)$ with
$\Lambda(b,m):=\tilde h(b;m)/G(b)$. Since
$\partial_b\Lambda=(\beta G-\tilde hG')/G^2$ vanishes at $b=b^*$ by
\eqref{eq:spG}, the movement of $b^*$ contributes nothing and
$\frac{d}{dm}\Lambda(b^*(m),m)=1/G(b^*)$. Hence $\partial_mU=G/G(b^*)$ and
\begin{equation}\label{eq:dmrho}
\partial_m\rho(v)=\frac{G(v)}{G(b^*)}-1\in(-1,0),
\qquad
\partial_m\rho'(v)=\frac{G'(v)}{G(b^*)}<0 ,
\end{equation}
for $v>b^*$; \eqref{eq:dcmrho} is \eqref{eq:dmrho} divided by $\delta$, together
with $\EE_v[e^{-\delta T_{b^*}}]=G(v)/G(b^*)$ for $v>b^*$.

\emph{Step 2: the implicit function theorem.} Write
$\Psi_d(z,m):=\rho'(z)F(z)-\rho(z)F'(z)$, so that $d^*$ solves $\Psi_d=0$ by
\eqref{eq:spd}. Both $\rho$ and $F$ are twice differentiable, so
$\partial_z\Psi_d=\rho''F-\rho F''$; multiplying by $\gamma^2z/2$ and adding and
subtracting $\kappa_\PP(\theta_\PP-z)\rho'F$ gives the identity
\[
\tfrac{\gamma^2z}{2}\bigl(\rho''F-\rho F''\bigr)
=F\,(\mathcal{L}-\delta)\rho-\kappa_\PP(\theta_\PP-z)\bigl(\rho'F-\rho F'\bigr),
\]
where $\mathcal{L}F=\delta F$ was used to eliminate $F''$. At $z=d^*$ the second
term vanishes, so
\begin{equation}\label{eq:dzPsi}
\partial_z\Psi_d(d^*)=\frac{2}{\gamma^2d^*}\,F(d^*)\,(\mathcal{L}-\delta)\rho(d^*).
\end{equation}
Note that \eqref{eq:dzPsi} holds \emph{at $d^*$ only}: away from a zero of
$\Psi_d$ the second term in the identity does not vanish and carries no
determined sign, so nothing here says $\Psi_d$ is monotone on either side of
$v_\rho$, and $v_\rho$ is in general neither a zero nor a stationary point of
$\Psi_d$.

It remains to sign \eqref{eq:dzPsi}. By Theorem~\ref{thm:entry}, $\varphi(d^*)$
maximises $H_\rho$, so $d^*$ maximises the ratio $\rho/F$, and
$(\rho/F)'=\Psi_d/F^2$ gives $(\rho/F)'(d^*)=0$, which is \eqref{eq:spd}. The
obstacle $\rho$ is $C^\infty$ on $(b^*,\infty)$ and $F$ is analytic and
positive, so $\rho/F$ is $C^2$ near the interior point $d^*$ and the
second-order condition for a maximum gives $(\rho/F)''(d^*)\le0$; since
$\Psi_d(d^*)=0$,
\[
(\rho/F)''(d^*)=\frac{\partial_z\Psi_d(d^*)}{F(d^*)^2},
\]
so $\partial_z\Psi_d(d^*)\le0$ and hence, by \eqref{eq:dzPsi} and $F>0$,
$(\mathcal{L}-\delta)\rho(d^*)\le0$. By Lemma~\ref{lem:rho} that function is
affine with slope $\beta(\kappa_\PP+\delta)<0$ and single root $v_\rho$, so it
is non-positive exactly on $[v_\rho,\infty)$: therefore $d^*\ge v_\rho$.

The inequality is strict, and the convexity structure already established
supplies it without further computation. By Lemma~\ref{lem:rho}, $H_\rho$ is
strictly convex on $(\varphi(0),\varphi(v_\rho)]$ and concave on
$[\varphi(v_\rho),0)$. A stationary point in the strictly convex part is a local
minimum, and a stationary $\varphi(v_\rho)$ is not a maximum either: since
$H_\rho'$ is strictly increasing to its left by convexity and vanishes there by
stationarity, $H_\rho'<0$ just to the left, putting higher values further left. So the stationary point
that is the maximum lies strictly inside the concave part, and $\varphi$ is
increasing: $d^*>v_\rho$. Consequently $\partial_z\Psi_d(d^*)<0$, and $d^*$ is
continuously differentiable in $m$ with
\begin{equation}\label{eq:ift}
\frac{\partial d^*}{\partial m}=-\frac{\partial_m\Psi_d(d^*)}{\partial_z\Psi_d(d^*)},
\qquad
\partial_m\Psi_d(d^*)
=\frac{G'(d^*)F(d^*)-G(d^*)F'(d^*)}{G(b^*)}+F'(d^*)
=F'(d^*)-\frac{W(d^*)}{G(b^*)},
\end{equation}
using \eqref{eq:dmrho}. Combining \eqref{eq:dzPsi} and \eqref{eq:ift} and
dividing by $\delta$ gives \eqref{eq:ddcm}, whose sign is that of
$\partial_m\Psi_d(d^*)$ because $\partial_z\Psi_d(d^*)<0$.

\emph{Step 3: the sign.} Since $\rho(b^*)=-k<0$ and $\rho/F$ attains a positive
maximum at $d^*$ --- positive because the entry region is non-empty --- we have
$\rho(d^*)>0$, and then $\rho'(d^*)=\rho(d^*)F'(d^*)/F(d^*)>0$ by \eqref{eq:spd}. Using
$F(d^*)/F'(d^*)=\rho(d^*)/\rho'(d^*)$,
\[
\frac{W(d^*)}{F'(d^*)}=G(d^*)-\frac{F(d^*)}{F'(d^*)}G'(d^*)
=\frac{\rho'(d^*)G(d^*)-\rho(d^*)G'(d^*)}{\rho'(d^*)} ,
\]
so $\partial_m\Psi_d(d^*)>0$ if and only if
$\rho'(d^*)G(d^*)-\rho(d^*)G'(d^*)<\rho'(d^*)G(b^*)$. With
$\rho=\Lambda^*G-\tilde h-k$ on $(b^*,\infty)$ --- the second branch of
\eqref{eq:rhoclosed}, with $\Lambda^*:=\Lambda(b^*,m)=\tilde h(b^*)/G(b^*)$ the
coefficient of Step 1 --- and using $\Lambda^*G(b^*)=\tilde h(b^*)$, the two
sides are
$-\beta G(d^*)+(\tilde h(d^*)+k)G'(d^*)$ and $\tilde h(b^*)G'(d^*)-\beta G(b^*)$, and
since $\tilde h(b^*)-\tilde h(d^*)=\beta(b^*-d^*)$ the inequality reads
\begin{equation}\label{eq:key}
-\beta\bigl(G(d^*)-G(b^*)\bigr)<\bigl(\beta(b^*-d^*)-k\bigr)G'(d^*).
\end{equation}
By the mean value theorem $G(d^*)-G(b^*)=G'(\hat v)(d^*-b^*)$ for some
$\hat v\in(b^*,d^*)$, so \eqref{eq:key} is
\[
-\beta\,(d^*-b^*)\bigl(G'(\hat v)-G'(d^*)\bigr)<-k\,G'(d^*).
\]
By complete monotonicity $G''>0$, so $G'$ is strictly increasing and
$\hat v<d^*$ gives $G'(\hat v)<G'(d^*)$: the left side is negative, while $k>0$ and
$G'(d^*)<0$ make the right side positive. The inequality holds strictly, so
$\partial_m\Psi_d(d^*)>0$ and $\partial d^*/\partial c_m>0$.

For $\beta>0$ the exercise region is $[b^*,\infty)$ and the entry region
$(0,d^*]$ with $d^*<b^*$, and Steps 1--3 run with $F$ and $G$ interchanged:
$\partial_m\rho=F/F(b^*)-1$, $\partial_m\Psi_d(d^*)=W(d^*)/F(b^*)+G'(d^*)$, and
\eqref{eq:key} becomes
$\beta(F(b^*)-F(d^*))>(\beta(b^*-d^*)-k)F'(d^*)$, that is
$\beta(b^*-d^*)\bigl(F'(\hat v)-F'(d^*)\bigr)>-kF'(d^*)$ with $\hat v\in(d^*,b^*)$.
Since $F''>0$ the left side is non-negative and the right side negative, so
$\partial_m\Psi_d(d^*)<0$ and $d^*$ falls as $c_m$ rises.
\end{proof}

\section{A state-dependent carrying charge}
\label{app:extras}

This appendix collects two extensions that the body sets aside. The first
gives the long position an entry rule, which Remark~\ref{rem:long} shows it
otherwise lacks. The second asks how much of the construction survives if the
variance risk premium is withdrawn altogether.

Everything so far rests on $\lambda\neq0$. By Corollary~\ref{cor:plam} the
timing problem disappears when the premium does, and by
Proposition~\ref{prop:carry} the premium is the entire content of
$\mathcal{A}\Phi$ once the spread and the carrying charge are netted into the
constant $s-c_m$. That is a
heavy load for one parameter to carry, and it is worth asking what happens if
the empirical support for it is withdrawn.

There is reason to ask. \citet{dewbecker2017} decompose the return to selling
variance by the \emph{horizon} of the variance being sold, and find it
concentrated at one and two months: claims on variance realised more than a
quarter ahead earn a premium they cannot distinguish from zero, and that remains
true out to the fourteen-year horizon their strips reach. A perpetual claim
loads on exactly the horizons at which they find nothing. Nor is the mismatch a
matter of calibration. Writing $\pi(u,v)=\EE^\QQ_v[v_u]-\EE^\PP_v[v_u]$ for the
per-horizon premium and averaging over the stationary law of $v$ under $\PP$,
under which $\EE[v]=\theta_\PP$, every two-measure affine specification gives
\begin{equation}\label{eq:piu}
\EE\bigl[\pi(u)\bigr]=(\theta_\QQ-\theta_\PP)\bigl(1-e^{-\kappa_\QQ u}\bigr),
\end{equation}
which is zero at $u=0$ and grows monotonically in magnitude to
$|\theta_\QQ-\theta_\PP|$.
Expression \eqref{eq:piu} involves neither $\kappa_\PP$ nor the link
\eqref{eq:link}, so no choice of $(\kappa_\QQ,\theta_\QQ)$, and no freeing of
the link, makes it decay. The model puts the premium at long horizons; the
measurement finds it only at short ones.

The response is not to argue with the measurement but to observe that $\lambda$
is not the only thing that can make $\tilde h$ non-constant.
Section~\ref{sec:setup} took the carrying charge to be the constant $c_m$, and
there is no reason it should be. A long variance position held against a book
that is structurally short volatility is worth more to that book the more
stressed it is; margin on a short position scales with the level of variance
rather than sitting at a fixed number of basis points. Both are statements that
the charge depends on the state. Take the simplest such dependence,
\begin{equation}\label{eq:cstate}
c(v)=c_m+c_1v ,
\end{equation}
with $c_1<0$ a hedging benefit that widens in stress and $c_1>0$ a financing
charge that rises with it.

Nothing in the construction breaks. The running reward of \eqref{eq:obj} becomes
$(\varepsilon-c_1)v-c_m$, which is still affine, so
Proposition~\ref{prop:reduce} holds verbatim --- the already accrued balance $A$
still carries the coefficient $\varepsilon$, since the charge applies to
variance not yet accrued --- and the carry identity \eqref{eq:carry} acquires
one term,
\begin{equation}\label{eq:carrystate}
\mathcal{A}\Phi(t,v)=\bigl[\varepsilon\lambda \gQ(t)-c_1\bigr]v+(s-c_m).
\end{equation}
On the perpetual contract the running reward is
$\varepsilon(v_u-K)-(c_m+c_1v_u)=(\varepsilon-c_1)v_u-\varepsilon K-c_m$: the
accrual and the state-dependent part of the charge are both linear in $v$ and
both accumulate under $\PP$, so they share one perpetuity \eqref{eq:perp} and
\[
R(v)=(\varepsilon-c_1)P_\PP(v)-\frac{\varepsilon K}{\delta}-\frac{c_m}{\delta},
\]
which is Proposition~\ref{prop:pcarry}'s $R$ with $\varepsilon$ replaced by
$\varepsilon-c_1$. Hence
\begin{equation}\label{eq:pbetastate}
\tilde h(v)=\varepsilon D(v)+c_1P_\PP(v)-\sbar+\frac{c_m}{\delta},
\qquad
\beta=-\frac{\varepsilon\lambda}{(\delta+\kappa_\QQ)(\delta+\kappa_\PP)}
+\frac{c_1}{\delta+\kappa_\PP},
\end{equation}
with $P_\PP$ the perpetuity \eqref{eq:perp}. The entry strike
still cancels, so Corollaries~\ref{cor:pgeom} and \ref{cor:pnoK} hold with
$\beta$ read from \eqref{eq:pbetastate}. What does not hold is
Corollary~\ref{cor:plam}: the premium is no longer the sole source of
optionality. Now $\beta=0$ if and only if
$c_1=\varepsilon\lambda/(\delta+\kappa_\QQ)$, and the timing problem degenerates
on that surface rather than on $\{\lambda=0\}$.

Two consequences are worth separating. The first concerns the accruing contract.
Since $\gQ(T)=0$, the slope in \eqref{eq:carrystate} tends to $-c_1$ as the swap
runs off, so near maturity the charge alone fixes the geometry whatever the
premium is. If $\varepsilon\lambda$ and $c_1$ share a sign and
$c_1/(\varepsilon\lambda)<\gQ(0)$ there is a unique $t^\ast$ at which
$\gQ(t^\ast)=c_1/(\varepsilon\lambda)$, and the exercise region changes sides
during the life of the contract --- a possibility the constant-charge model does
not admit at all.

The second is the point of the subsection. Set $\lambda=0$. There is then no
premium at any horizon, so \eqref{eq:piu} vanishes identically and there is
nothing left for \citet{dewbecker2017} to contradict; $D\equiv0$, and
\begin{equation}\label{eq:nopremium}
\tilde h(v)=c_1P_\PP(v)-\sbar+\frac{c_m}{\delta},
\qquad
\beta=\frac{c_1}{\delta+\kappa_\PP} .
\end{equation}
The position $\varepsilon$ has disappeared: with no premium the two sides are
symmetric and the charge alone decides both the geometry and the thresholds. The
stopping problem is nonetheless non-trivial, and the two-stage structure of
\eqref{eq:Jentry} goes through unchanged. Table~\ref{tab:statedep} solves it
on the physical law of \eqref{eq:cal7}.

\begin{table}[t]\centering
\begin{tabular}{rrrll}
\toprule
$c_1$ & $c_m$ & $c(v)$ at $p1/p50/p99$ & enter & unwind\\
\midrule
$-0.10$ & $\phantom{-}\phantom{0}40$ bp & $+35\,/\,+4\,/\,-77$ bp
  & above $23.40\%\ (p98.9)$ & never\\
$-0.20$ & $\phantom{-}\phantom{0}80$ bp & $+69\,/\,+8\,/\,-154$ bp
  & above $22.54\%\ (p96.5)$ & below $18.46\%\ (p0.9)$\\
$-0.30$ & $\phantom{-}120$ bp & $+104\,/\,+12\,/\,-230$ bp
  & above $22.20\%\ (p94.7)$ & below $18.63\%\ (p3.1)$\\
$-0.50$ & $\phantom{-}240$ bp & $+213\,/\,+60\,/\,-344$ bp
  & above $22.46\%\ (p96.2)$ & below $19.32\%\ (p24.1)$\\
\midrule
$+0.20$ & $-\phantom{0}80$ bp & $-69\,/\,-8\,/\,+154$ bp
  & below $18.62\%\ (p2.9)$ & above $22.57\%\ (p96.7)$\\
$+0.30$ & $-120$ bp & $-104\,/\,-12\,/\,+230$ bp
  & below $18.82\%\ (p7.1)$ & above $22.31\%\ (p95.4)$\\
\bottomrule
\end{tabular}
\caption{Entry and unwind rules with \emph{no variance risk premium}
($\lambda=0$, so $D\equiv0$), driven entirely by the state-dependent carrying
charge \eqref{eq:cstate}. The physical law is that of \eqref{eq:cal7}, with
$\delta=0.5$ and $\sbar=\sebar=20$ basis points; since $\lambda=0$ the perpetual fair
rate is the physical one. Thresholds are perpetual rates with percentiles of the
stationary law of $v$ in brackets, and the third column reports the charge in
basis points a year at the first percentile, the median and the ninety-ninth.
The first row is the degenerate regime of
Proposition~\ref{prop:holdforever}: the slope is too shallow to pay for the
spread, and the position once entered is held to termination.}
\label{tab:statedep}
\end{table}

Two features of the table matter more than the individual numbers. The
\emph{level} of the charge is almost irrelevant: at the median state it never
exceeds $60$ basis points a year in absolute value, and in five of the six rows
it does not exceed $12$. What creates the optionality is the \emph{slope}, and
\eqref{eq:nopremium} says so exactly --- $\beta$ depends on $c_1$ alone, and
$c_m$ enters only the intercept, exactly as in \eqref{eq:htD}. And
the round trips are of a usable length: at $c_1=+0.30$, $c_m=-120$ basis points,
simulation under $\PP$ from the stationary law puts the mean time to entry at
$1.5$ years and the mean holding period at $4.2$; at $c_1=-0.50$, $c_m=240$
basis points, $3.9$ years and $1.3$ --- the magnitudes of the premium-driven
round trips of Section~\ref{sec:working}, produced with no premium at all.

The honest reading of this is narrower than it may appear. The coefficient $c_1$
is a preference parameter, in the class of $c_0$ and not in the class of
$\lambda$ or of $\delta$: it describes the trader's book rather than the market's
prices, and it cannot be read off an option surface. The subsection therefore
trades an empirical claim for a behavioural one. It does not assert that a
premium of the required shape exists; it says that if a hedging demand of this
shape exists, then here is when to trade against it, and it prices how large the
state-dependence must be before the two-threshold structure appears at all --- a
swing of roughly two hundred basis points a year across the middle $98\%$ of the
law, which is an order of magnitude a risk manager can be asked to confirm or
deny. That is a weaker claim than the one Section~\ref{sec:carry} makes. It is
also the one that survives \citet{dewbecker2017}.

Finally, the generalisation does not rescue the entry threshold, and it is worth
recording how it fails. For the short at the three-point premium of
\eqref{eq:cal7}, a financing charge rising with variance is the realistic
specification, $c_1>0$, and it makes matters worse. Taking the holding cost
that makes the entry rule as reachable as it can be at each slope, subject to
Definition~\ref{def:nondeg} at $\eta$, the entry threshold is the $93$rd
percentile at $c_1=0$, the $97$th at $c_1=0.15$ --- a charge of $54$ basis
points a year at the median state --- and the $99$th at $c_1=0.20$; from
$c_1=0.25$ to $c_1=0.40$ no holding cost admits a non-degenerate pair at all.

Past that band the geometry inverts. At
$c_1=\varepsilon\lambda/(\delta+\kappa_\QQ)=0.323$, or $117$ basis points a year
at the median, $\beta$ changes sign by \eqref{eq:pbetastate} and the entry
region becomes the lower set: the short is then sold in calm markets and bought
back into a spike. The entry threshold does become reachable there --- the
$15$th percentile at $c_1=0.50$, the $30$th at $c_1=0.70$ --- but the holding
cost that achieves it has turned negative, from $-84$ to $-190$ basis points a
year. What buys the reachable entry rule on that side is a subsidy, not a
charge, and a subsidy is what makes a long worth opening in the first place.
As a charge, $c_1$ leaves the asymmetry of
Section~\ref{sec:idle} where it found it: a statement about how thin the
tails of the physical law are, which a state-dependent charge does not thicken.

One caveat attaches to all of these figures, and to the interval of
Figure~\ref{fig:window}. They are screened by the reachability test \eqref{eq:reach} at $\eta=0.5\%$, and
the screen is doing more work than a tolerance should. Raising it to $5\%$ --- still a rule
that fires on one observation in twenty --- closes the window of
Figure~\ref{fig:window} entirely below a premium of about three and a half
volatility points: at the three-point premium of \eqref{eq:cal7} the $93$ basis
points of admissible carrying charge become none. Where the window survives it
is narrow, $14$ basis points at four points against $159$ at $\eta=0.5\%$. The round trips this paper reports are therefore usable in the
sense that the physical law reaches both thresholds, not in the stronger sense
that it reaches them often.

\end{document}